\documentclass[12pt]{article}%
\usepackage{amsfonts}
\usepackage{amsmath}
\usepackage{amssymb}
\usepackage{graphicx}%
\providecommand{\U}[1]{\protect\rule{.1in}{.1in}}
\newtheorem{theorem}{Theorem}

\newtheorem{definition}[theorem]{Definition}

\newtheorem{lemma}[theorem]{Lemma}

\newtheorem{proposition}{Proposition}

\newenvironment{proof}[1][Proof]{\noindent \textbf{#1.} }{\  \rule{0.5em}{0.5em}}

\begin{document}

\title{A Dynamic Analysis of Nash Equilibria in Search Models with Fiat
Money\thanks{Correspondence to: maurizio.iacopetta@sciencespo.fr. We are
grateful to\ conference and seminar participants at the 2017 Summer Workshop
on Money, Banking, Payments, and Finance at the Bank of Canada, the IV AMMCS
International Conference, Waterloo, Canada, Luiss University (Rome),
University of G\"{o}ttingen, and the School of Mathematics, Georgia Tech, for
useful comments. All remaining errors are ours.}}
\author{Federico Bonetto, School of Mathematics, Georgia Tech
\and Maurizio Iacopetta, OFCE (Sciences-Po) and Skema Business School}
\date{}
\maketitle

\begin{abstract}
We study the rise in the acceptability fiat money in\ a Kiyotaki-Wright
economy by developing a method that can determine dynamic Nash equilibria for
a class of search models with genuine heterogenous agents. We also address
open issues regarding the stability properties of pure strategies equilibria
and the presence of multiple equilibria. Experiments illustrate the liquidity
conditions that favor the transition from partial to full acceptance of fiat
money, and the effects of inflationary shocks on production, liquidity, and trade.


\end{abstract}

\section{Introduction}

One central question of monetary economics is how an object that does not
bring utility \textit{per se} is accepted as a mean of payment. It is well
understood that the emergence of money depends on trust and coordination of
believes. While some recognize this observation and simply assume that money
is part of the economic system, others have tried to explain the acceptance of
money as the result of individuals' interactions in trade and production
activities.\footnote{Economic textbooks sometimes alert readers that the
acceptance of fiat money cannot be taken for granted. For instance, Mankiw
(2006, p. 644-45) observes that in the late 1980s, when the Soviet Union was
breaking up, in Moscow some preferred cigarettes to rubles as means of
payment.} Among the best-known attempts to formalize the emergence of money in
decentralized exchanges is Kiyotaki and Wright (1989) (henceforth, KW). The
static analysis in KW provides important insights on how specialization in
production, the technology of matching, and the cost of holding commodities,
condition the emergence of monetary equilibria. Nevertheless, it leaves
important issues open. First, one would like to know if and how convergence to
a particular long run equilibrium occurs from an arbitrary initial state of
the economy. Historical accounts describe different patterns that societies
followed in adopting objects as means of payment.\footnote{For a classic
review of the rise of early means of payments see Quiggin (1949). For the
institutional and historical conditions that favored the dissemination of fiat
money see Goetzmann (2016).} What are the dynamic conditions that lead
individuals in a KW economy to accept commodity or fiat money? Second, static
analysis gives little guidance about the short run consequences of a shock
that causes, for instance, a sudden rise of inflation. How does the degree of
acceptability of commodity and fiat money change with inflation?

The determination of dynamic equilibria in a KW environment is challenging. In
an effort to improve its tractability, new classes of monetary search models
have been proposed. These have incorporated some features of centralized
exchanges but have also eliminated others, most notably the genuine
heterogeneity across individuals and goods, and the storability of goods (see
Lagos et al. (2017), for a recent review). Restoring these features turns out
to be a useful exercise for characterizing the rise of money as a dynamic phenomenon.

The study of money acceptance in a KW environment requires a departure from
the conventional set of tools employed to characterize the dynamics of an
economy with centralized markets. Our method combines Nash's (1950) definition
of equilibrium with Perron's iterative approach to prove the stable manifold
theorem (see, among others, Robinson, 1995). This is the first work, to our
knowledge, that shows how to determine pure strategies dynamic Nash equilibria
in a KW\ search environment with fiat money. Previous works on the subject
considered economies without fiat money, and often assumed bounded
rationality.\footnote{See, for instance, the works of Marimon et al. (1990)
and Ba\c{s}\c{c}\i\ (1999) with intelligent agents, and of Brown (1996) and
Duffy and Ochs (1999, 2002) with controlled laboratory experiments. Matsuyama
et al. (1993), Wright (1995), Luo (1999) and Sethi (1999) use evolutionary
dynamics. Kehoe, Kiyotaki, and Wright (1993) show that under mixed strategies
equilibria could generate cycles, sunspots, and other non-Markovian
equilibria. Renero (1998), however, proves that it is impossible to find an
initial condition from which an equilibrium pattern converges to a mixed
strategy steady state equilibrium. Oberfield and Trachter (2012) find that, in
a symmetric environment, as the frequency of search increases, cycles and
multiplicity in mixed strategy tend to disappear. More recently, Iacopetta
(2018) studies dynamic Nash equilibria in a KW environment with no fiat
money.}

Steady state results echo those of inventory-theoretic models of money (e.g.,
Baumol 1952, Tobin 1956, and Jovanovic 1982): for instance, higher levels of
seignorage may induce some to keep commodities in the inventory instead of
accepting money, as a way to minimize the odds of being hit by a seignorage
tax. The dynamic analysis, however, generates novel results: it shows how
changes in the liquidity of assets other than money can alter the proportion
of individuals who accept fiat money in transactions. For instance, it reveals
that an economy that converges to a long run equilibrium in which all prefer
fiat money to all types of commodities (full acceptance), may go through a
phase in which only a fraction of individuals do so (partial
acceptance).\footnote{Shevchenko and Wright (2004) also study partial
acceptability in pure strategies. They focus, however, on steady state
analysis.}

The remainder of the paper is organized as follows. Section 2 describes the
economic environment, characterizes the evolution of the distribution of
inventories and money, and defines a Nash equilibrium. Section 3 overviews
steady state Nash equilibria for some specifications of the model. Section 4
presents a methodology to determine Nash equilibria. Section 5 illustrates the
acceptability of money and discusses multiple steady states through numerical
experiments. Section 6 contains welfare considerations. Section 7 has few
remarks about future research. Appendix A contains proofs and mathematical
details omitted in the main text. Appendix B explains how the stable manifold
theorem is related to our solution algorithm.

\section{The Model}

This section describes the economic environment, characterizes the evolution
of the distribution of inventories and money, and defines a Nash equilibrium.

\subsection{The Environment}

The model economy is a generalization of that described in KW. There are four
main differences. First, to facilitate the analysis of the dynamics, time is
continuous. Second, the model is extended to deal with seignorage, following
the approach devised by Li (1994, 1995): government agents randomly confiscate
money holders of their balances and use the proceedings to purchase
commodities. Third, as in Wright (1995) agents are not necessarily equally
divided among the three types. Fourth, as in Lagos et al. (2017), we obtain
the type of equilibria that emerge in the Model B of KW by reshuffling the
ordering of the storage costs across the three types of goods rather than
altering the patterns of specialization in production.

The economy is populated by three types of infinitely lived agents; there are
$N_{i}$ individuals of type $i$, with $i=1,2,3$, where $N_{i}$ is a very large
number. The total size of the population is $N=N_{1}+N_{2}+N_{3}$ and the
fraction of each type is denoted with $\theta_{i}=\frac{N_{i}}{N}$. A type $i$
agent consumes only good $i$ and can produce only good $i+1$ (modulo 3).
Production occurs immediately after consumption. Agent $i$'s instantaneous
utility from consuming a unit of good $i$ and the disutility of producing good
$i+1$ are denoted by $U_{i}$ and $D_{i}$, respectively, with $U_{i}>D_{i}>0$,
and their difference with $u_{i}=U_{i}-D_{i}$. The storage cost of good $i$ is
$c_{i}$, measured in units of utility.\footnote{There is no restriction on the
sign of $c_{i}$. A negative storage cost is equivalent to a positive return.}
In addition to the three types of commodity there is a fourth object, called
money and denoted by $m$, that does not bring utility \textit{per se}: it only
serves as a means of transaction. Fiat money is indivisible. Denoting with $M$
the fraction of the population holding fiat money, the total quantity of fiat
money is $Q=NM$. There is no cost for storing money. At each instant in time,
an individual can hold one and only one unit of any type $i$ good or one unit
of money.\footnote{The assumption that an individual can have either 0 or 1
unit of an asset greatly simplifies the analysis and makes more transparent
the decision about the acceptance of money. There is a significant amount of
work with asset-holding restrictions. See, among others, Diamond (1982),
Rubinstein and Wolinsky (1987), Cavalcanti and Wallace (1999), and Duffie et
al. (2005).}

The discount rate is denoted by $\rho>0$. A pair of agents is randomly and
uniformly chosen from the population to meet for a possible trade. The
matching process is governed by a Poisson process with exogenous arrival rate
$\frac{\alpha N}{2}$, where $\alpha>0$ -- i.e. there is a constant returns to
scale matching technology. Hence, after a pair is formed, the expected waiting
time for the next pair to be formed is $\frac{2}{\alpha N}$. A bilateral trade
occurs if, and only if, it is mutually agreeable. Agent $i$ always accepts
good $i$ but never holds it because, provided that $u_{i}$ is sufficiently
large, there is immediate consumption (see KW, Lemma 1, p. 933). Therefore,
agent $i$ enters the market with either one unit of good $i+1$, or $i+2$, or
with one unit of $m$.

We introduce seignorage as Li (1994, 1995): The government extracts seignorage
revenue from money holders in the form of a money tax -- this device has been
used by many others, including in the recent work of Deviatov and Wallace
(2014). In particular, government agents meet and confiscate money from money
holders. The arrival rate of a government agent for a money holder is
$\delta_{m}$. A money holder of type $i$, whose unit of money is confiscated,
returns to the state of production without consumption, produces a new
commodity $i+1$, and incurs a disutility $D_{i}$. With the proceedings of the
tax revenue the government purchases goods from commodity holders. These
encounters are governed by a Poisson process with arrival rate $\delta_{g}$.
The government runs, on average, a balanced budget. This requires that
$\delta_{m}M=\delta_{g}(1-M),$ implying that $\delta_{g}=\frac{\delta_{m}%
M}{1-M}$. We allow the government to alter the rate of seignorage $\delta_{m}%
$, but, in order to simplify the dynamic analysis, the government does not
change the real balances in circulation -- i.e. the initial level of fiat
money is given and does not change over time. We will compare, however, steady
state equilibria of economies with different levels of $M$.

\subsection{Distribution of Commodities and Fiat Money}

Let $p_{i,j}(t)$ denote the proportion of type $i$ agents that hold good $j$
at time $t$. A type $i$ with good $j$ has to decide during a meeting whether
to trade $j$ for $k$, where $j,k=i+1$, $i+2$, or $m$ (henceforth we no longer
mention the ranges of the indices $i,$ $j$ and $k$, unless needed to prevent
confusion). Agent $i$'s decision in favor of trading $j$ for $k$ is denoted by
$s_{j,k}^{i}=1$, and that against it by $s_{j,k}^{i}=0$. The evolution of
$p_{i,j}$, for a given set of strategies, $s_{j,k}^{i}(t)$, is governed by a
system of differential equations\footnote{We assume that the influence of any
particular individual on the system is negligible. Eqs.
\eqref{pi}-\eqref{pizero} should be interpreted as the limit of the stochastic
evolution of the inventories of an economy with a finite number of agents. See
Araujo (2004) and Araujo et al. (2012) for a discussion of the system's
properties of similar economies with a finite number of agents.} (the time
index is dropped):
\begin{align}
\dot{p}_{i,i+1}=  &  \alpha\left\{  \sum_{i^{\prime}}\sum_{k}p_{i,k}%
p_{i^{\prime},i+1}s_{k,i+1}^{i}s_{i+1,k}^{i^{\prime}}+\sum_{i^{\prime}}%
p_{i,k}p_{i^{\prime},i}s_{i,k}^{i^{\prime}}-\sum_{i^{\prime}}\sum_{k}%
p_{i,i+1}p_{i^{\prime},k}s_{i+1,k}^{i}s_{k,i+1}^{i^{\prime}}\right\}
\nonumber\\
+  &  \delta_{m}p_{i,m}-\delta_{g}p_{i,i+1}\label{pi}\\
\dot{p}_{i,i+2}=  &  \alpha\left\{  \sum_{i^{\prime}}\sum_{k}p_{i,k}%
p_{i^{\prime},i+2}s_{k,i+2}^{i}s_{i+2,k}^{i^{\prime}}-\sum_{i^{\prime}}%
\sum_{k}p_{i,i+2}p_{i^{\prime},k}s_{i+2,k}^{i}s_{k,i+2}^{i^{\prime}}\right\}
-\delta_{g}p_{i,i+2}\label{pitwo}\\
\dot{p}_{i,m}=  &  \alpha\left\{  \sum_{i^{\prime}}\sum_{k}p_{i,k}%
p_{i^{\prime},m}s_{k,m}^{i}s_{m,k}^{i^{\prime}}-\sum_{i^{\prime}}\sum
_{k}p_{i,m}p_{i^{\prime},k}s_{m,k}^{i}s_{k,m}^{i^{\prime}}\right\} \nonumber\\
-  &  \delta_{m}p_{i,m}+\delta_{g}(p_{i,i+1}+p_{i,i+2}). \label{pizero}%
\end{align}
Focusing on the top equation ($\dot{p}_{i,i+1}$), the first two sums inside
the brackets, starting from the left, account for events that lead to an
increase in the share of individuals $i$ with $i+1$. Specifically, the term
$p_{i,k}p_{i^{\prime},i+1}$ is the probability that a type $i$ with good $k$
meets a type $i^{\prime}$ with good $i+1$, and $s_{k,i+1}^{i}s_{i+1,k}%
^{i^{\prime}}$ calculates their willingness to swap goods: if they both agree
to trade, $s_{k,i+1}^{i}s_{i+1,k}^{i^{\prime}}=1$; if one of the two does not,
$s_{k,i+1}^{i}s_{i+1,k}^{i^{\prime}}=0$. The term $p_{i,k}p_{i^{\prime},i}$
considers the residual case in which type $i$ with good $k$ meets a type
$i^{\prime}$ with good $i$. Because type $i$ always accepts good $i$, trade
take places as long as $s_{i,k}^{i^{\prime}}=1$. The third sum accounts for
events that cause a decline in $p_{i,i+1}$. Finally, the last two terms,
$\delta_{m}p_{i,m}$ and $\delta_{g}p_{i,i+1}$, measure the overall amount of
fiat money the government confiscates from money holders ($p_{i,m}$) and the
amount of goods $i+1$ it buys from type $i$ agents -- $\delta_{m}$ and
$\delta_{g}$ are the government's Poisson rates of intervention, respectively.

The extended form of eqs. \eqref{pi}-\eqref{pizero} consists of nine
non-autonomous non-linear differential equations in the nine unknowns
$p_{i,j}(t)$, for $i=1,2,3,$ and $j=i+1$, $i+2$, $m$ (they are non-autonomous
equations because $s_{j,k}^{i}(t)$ depends on $t$). Nevertheless, because
$p_{i,i}(t)=0$,
\begin{equation}
p_{i,i+1}(t)+p_{i,i+2}(t)+p_{i,m}(t)=\theta_{i}. \label{mui}%
\end{equation}
Another restriction comes from the following accounting relationship:
\begin{equation}
p_{1,m}(t)+p_{2,m}(t)+p_{3,m}(t)=M. \label{M}%
\end{equation}
Therefore, in \eqref{pi}-\eqref{pizero} there are only five independent
equations, and the state of the economy can be represented by the
five-dimensional vector ${\mathbf{p}}(t)=(p_{1,2}(t),p_{2,3}(t),p_{3,1}%
(t),p_{1,m}(t),p_{2,m}(t))$.

Let $\Omega$ be the set of $p^{i}_{j,k}(t)$ that satisfies \eqref{mui} and
\eqref{M}. For any sets of strategies $s_{j,k}^{i}(t)$, the solution of
\eqref{pi}--\eqref{pizero} maps $\Omega$ into itself. Because $\Omega$ is
compact and convex, the system \eqref{pi}--\eqref{pizero} admits at least one
fixed point for any constant sets of strategies $s_{j,k}^{i}$. Proposition
\ref{pro:stab-no-money} shows that in the simple scenario with no fiat money
($M=0$), the fixed point is unique and globally attractive. Section
\ref{numeri} studies the uniqueness and global attractiveness of the fixed
point numerically for an economy with fiat money.

\subsection{Value Functions}

The system of equations \eqref{pi}-\eqref{pizero} specifies the evolution of
the economy for a given set of strategies $s_{j,k}^{i}(t)$. We turn now to the
individuals' decisions about trading strategies. Denote with $\sigma_{j,k}%
^{i}(t)$ the strategies of a particular agent of type $i$, $a_{i}$, who takes
for given the strategies of the rest of the population, $s_{j,k}^{i}(t)$,
including agents of her own type, and who knows the initial state
${\mathbf{p}}(0)$. Let $V_{i,j}(t)$ be the integrated expected discounted flow
of utility from time $t$ onward of this particular agent $a_{i}$ with good $j$
at time $t$. Then,
\begin{equation}
V_{i,j}(t)=\int_{t}^{\infty}e^{-\rho(\tau-t)}\sum_{l}\pi_{l,j}^{i}%
(\tau,t)v_{i,l}({\mathbf{p}}(\tau))d\tau\text{,} \label{Vf}%
\end{equation}
where $\pi_{l,j}^{i}(\tau,t)$ is the probability that $a_{i}$, who carries
good $j$ at time $t$ and plays strategy $\sigma_{j,k}^{i}(t)$, holds good $l$
at time $\tau\geq t$, and $v_{i,l}({\mathbf{p}})$ is $a_{i}$'s flow of
utility, net of storage costs, associated to the distribution of holdings
${\mathbf{p}}$. Observe that $\pi_{l,j}^{i}(\tau,t)$ and $v_{i,l}({\mathbf{p}%
})$ both depend on what other individuals do, that is they are affected by
$s_{j^{\prime},k^{\prime}}^{i^{\prime}}(\tau)$, $\sigma_{j^{\prime},k^{\prime
}}^{i}(\tau)$ -- where $i^{\prime}=1,2,3$, and $j^{\prime},k^{\prime
}=i+1,i+2,m$. Because $\rho>0$ and $v_{i,l}({\mathbf{p}})$ is bounded, the
integral in (\ref{Vf})\ is well defined, for any $s_{j,k}^{i}(t)$ and any
${\mathbf{p}}(t)$. Appendix A contains the expression of $v_{i,l}({\mathbf{p}%
})$ and the evolution of $\pi_{l,j}^{i}(\tau,t)$. It also shows the following
result about the evolution of $V_{i,j}(t)$.

\begin{proposition}
\label{Vidot} The evolution of $V_{i,j}$ in (\ref{Vf})\ satisfies (time index
is dropped):
\begin{equation}
(\alpha+\delta_{k}+\rho)V_{i,j}=\dot{V}_{i,j}+\phi_{i,j}\text{,} \label{vidot}%
\end{equation}
where $\delta_{k}=\delta_{g}$ and
\begin{align}
\phi_{i,j}=  &  \alpha\left\{  \sum_{i^{\prime}}\sum_{k\not =i}p_{i^{\prime
},k}\sigma_{j,k}^{i}s_{k,j}^{i^{\prime}}V_{i,k}+\sum_{i^{\prime}}p_{i^{\prime
},i}s_{i,j}^{i^{\prime}}(V_{i,i+1}+u_{i})+\sum_{i^{\prime},k}p_{i^{\prime}%
,k}(1-\sigma_{j,k}^{i}s_{k,j}^{i^{\prime}})V_{i,j}\right\}  +\nonumber\\
+  &  \delta_{g}V_{i,m}-c_{j}\text{,} \label{phibis}%
\end{align}
when $j=i+1$ or $i+2$, and $\delta_{k}=\delta_{m}$ and%

\begin{align}
\phi_{i,m}=  &  \alpha\left\{  \sum_{i^{\prime}}\sum_{k\not =i}p_{i^{\prime
},k}\sigma_{m,k}^{i}s_{k,m}^{i^{\prime}}V_{i,k}+\sum_{i^{\prime}}p_{i^{\prime
},i}s_{i,m}^{i^{\prime}}(V_{i,i+1}+u_{i})+\sum_{i^{\prime},k}p_{i^{\prime}%
,k}(1-\sigma_{m,k}^{i}s_{k,m}^{i^{\prime}})V_{i,m}\right\}  +\nonumber\\
+  &  \delta_{m}(V_{i,i+1}-D_{i}), \label{phizero}%
\end{align}
when $j=m$.
\end{proposition}

\begin{proof}
See Appendix A.\bigskip
\end{proof}

The first sum in \eqref{phibis}, counting from the left, is the expected flow
of utility of agent $a_{i}$ with good $j$ conditional on meeting an agent
$i^{\prime}$ who carries a good $k\not =i$. Such a meeting occurs with
probability $p_{i^{\prime},k}$, and trade follows if $\sigma_{k,j}^{i}%
s_{j,k}^{i^{\prime}}=1$. In such a case, $a_{i}$ leaves the meeting with good
$k$ (i.e. with continuation value $V_{i,k}$). Similarly, the second sum
accounts for $a_{i}$'s expected flow of utility, conditional on meeting an
agent $i^{\prime}$ who carries good $i$. The last sum refers to meetings in
which no trade occurs, in which case $a_{i}$ is left with good $j$. The term
$\delta_{g}V_{i,m}$ is the expected continuation value for meeting government
agents who buy $a_{i}$'s good $j$ using fiat money, and $c_{j}$ is the cost of
storage. Because the term $\delta_{g}V_{i,m}$ appears in \eqref{phibis} for
both $j=i+1$ and $j=i+2$, the level of the seignorage does not directly affect
the optimal response $\sigma_{k,j}^{i}(t)$ -- it may do so only indirectly
through ${\mathbf{p}}(t)$.

\subsection{Best Response and Nash Equilibrium}

Agent $a_{i}$'s best response to the set of strategies of the other agents,
$s_{j,k}^{i}(t)$, is a set of strategy $\sigma_{j,k}^{i}(t)$ that maximizes
her expected flow of utility in (\ref{Vf}):
\begin{equation}
V_{i,j}\left(  t,\left\{  \sigma_{k,l}^{i}(\tau)\right\}
_{\substack{k,l=i+1,i+2,m\\\tau>t}}\right)  =\sup_{\tilde{\sigma}_{k,j}^{i}%
}V_{i,j}\left(  t,\left\{  \tilde{\sigma}_{k,l}^{i}(\tau)\right\}
_{\substack{k,l=i+1,i+2,m\\\tau>t}}\right)  \text{, for }\forall t\text{ and
}\forall j. \label{sup}%
\end{equation}
A useful characterization of the best response function is the following:

\begin{proposition}
\label{best} Let $\Delta_{j,k}^{i}(t)\equiv V_{i,j}(t)-V_{i,k}(t)$. The
strategy $\sigma_{j,k}^{i}(t)$ is a best response of $a_{i}$ to $s_{j,k}%
^{i}(t)$, for a given ${\mathbf{p}}(0)={\mathbf{p}}_{0}$, if and only if%

\begin{equation}
\sigma_{j,k}^{i}(t)=%
\begin{cases}
1 & \text{if }\Delta_{j,k}^{i}(t)<0\\
0 & \text{if }\Delta_{j,k}^{i}(t)>0\\
0.5 & \text{if }\Delta_{j,k}^{i}(t)=0\,
\end{cases}
. \label{consiste}%
\end{equation}

\end{proposition}

\begin{proof}
See Appendix A.
\end{proof}

When $a_{i}$ is indifferent between trading $j$ for $k$ with $j\neq k$, the
tie-breaking rule is that $\sigma_{j,k}^{i}(t)=0.5$.\footnote{There is only in
a finite set of isolated times $t_{l}$, $l=1,\ldots,L$ \ when $\Delta
_{j,k}^{i}$ can change sign. Therefore $s_{j,k}^{i}$ is discontinuous at
$t_{l}$. The value of $s_{j,k}^{i}(t_{l})$ at such points of discontinuity is
not relevant. Shevchenko and Wright (2004) make a similar observation.}
Clearly, \eqref{consiste} implies that $\sigma_{j,k}^{i}=1-\sigma_{k,j}^{i}$.
We also set $\sigma_{j,j}^{i}=0$, that is, $a_{i}$ never trades $j$ for $j$.
Therefore, the set of strategy of agent $a_{i}$ can be represented by
$\boldsymbol{\sigma}^{i}(t)=(\sigma_{i+1,m}^{i}(t),\sigma_{i+2,m}%
^{i}(t),\sigma_{i+1,i+2}^{i}(t))$, a piecewise continuous function from
$\mathbb{R}^{+}$ into $\Sigma
=\{(1,1,1),(1,0,1),(1,1,0),(0,1,0),(0,0,1),(0,0,0)\}$. Although agent $a_{i}$
has eight possible trading choices at each point in time, a simple
transitivity trading rule (for instance, if $\sigma_{2,3}^{1}=0$ and
$\sigma_{2,m}^{1}=1$, then it must also be that $\sigma_{3,m}^{1}=1$) reduces
her choices to the six contained in $\Sigma$ -- this applies to any type
$i=1,2,3$. We call $\boldsymbol{\sigma}(t)=(\boldsymbol{\sigma}^{1}%
(t),\boldsymbol{\sigma}^{2}(t),\boldsymbol{\sigma}^{3}(t))\in\Sigma^{3}$ the
best responses of the three particular agents $a_{i}$. Similarly,
$\mathbf{s}(t)=(\mathbf{s}^{1}(t),\mathbf{s}^{2}(t),\mathbf{s}^{3}%
(t))\in\Sigma^{3}$ denotes agents' symmetric strategies, with $\mathbf{s}%
^{i}(t)=(s_{i+1,m}^{i}(t),s_{i+2,m}^{i}(t),s_{i+1,i+2}^{i}(t))\in\Sigma$.
Next, following Nash (1950), we define an equilibrium by means of a function
$\boldsymbol{\sigma}=\boldsymbol{\mathcal{B}}({\mathbf{s}})$ that associates
the best response $\boldsymbol{\sigma}(t)$ to a set of strategies
$\mathbf{s}(t)$. The function $\boldsymbol{\mathcal{B}}$ transforms a
piecewise continuous function $\mathbf{s}:\mathbb{R}^{+}\rightarrow\Sigma^{3}$
into another piecewise continuous function $\boldsymbol{\sigma}%
=\boldsymbol{\mathcal{B}}({\mathbf{s}}):\mathbb{R}^{+}\rightarrow\Sigma^{3}$.

\begin{definition}
[Nash Equilibrium]Given an initial distribution ${\mathbf{p}}_{0}$, a set of
strategies ${\mathbf{s}}^{\ast}$ is a Nash equilibrium if it is a fixed point
of the map $\boldsymbol{\mathcal{B}}$:
\begin{equation}
{\mathbf{s}}^{\ast}=\boldsymbol{\mathcal{B}}({\mathbf{s}}^{\ast}).
\label{fixed}%
\end{equation}

\end{definition}

This definition equilibrium requires, therefore, that $\boldsymbol{\sigma}$,
the best response to the set of strategies ${\mathbf{s}}$, be equal to
${\mathbf{s}}$.

A general proof of the existence of such an equilibrium, for any given
${\mathbf{p}}_{0}$, cannot be obtained with the standard fixed-point argument
based on Kakutani or Brouwer theorems applied to finite games. These theorems
would require the best response function to be a continuous map on a convex
and compact set. Compactness, however, cannot be verified in our infinite time
horizon set up. Nevertheless, Proposition \ref{closeup} in Section (4) states
that sometimes the existence of Nash equilibria can be established
analytically near Nash steady states. Section \ref{dyna} constructs Nash
equilibria numerically, even when their existence cannot be established analytically.

\section{Overview of Steady States}

\label{steady}

We begin by considering an economy with no fiat money ($M=0$) and $\theta
_{i}=\frac{1}{3}$. Since the first two rows of $\mathbf{s}$ refer to the
acceptance of fiat money -- recall that the rows of $\mathbf{s}$ are
associated to objects and the columns to types -- we focus in what follows on
the entries of its third row, $\mathbf{s}_{3}$, shows how agents $i$ order
good $i+1$ and good $i+2$. For instance, when $\mathbf{s}_{3}=(0,1,0)$, type 2
trade $i+1$ for $i+2$ (i.e. 3 for 1), whereas types 1 and 3 do not. In
addition, because $p_{1,m}=p_{2,m}=0$, it is convenient to shorten
${\mathbf{p}}$ into $\mathbf{\hat{p}}=(p_{1,2}$, $p_{2,3}$, $p_{3,1})$.

Assume $c_{1}<c_{2}<c_{3}$ (model A of KW). There are eight possible
combinations of (pure)\ strategies. Two of them are\ Nash equilibria:
\begin{equation}
\mathbf{s}_{3}\mathbf{=}(0,1,0)\text{ with }\mathbf{\hat{p}}=\frac{1}%
{3}\left(  1,\,\frac{1}{2},\,1\right)  ,\text{ } \label{F}%
\end{equation}
if%
\begin{equation}
\frac{c_{3}-c_{2}}{u_{1}\alpha}>p_{3,1}-p_{2,1}=\frac{1}{6}\text{,} \label{CF}%
\end{equation}
and%
\begin{equation}
\mathbf{s}_{3}\mathbf{=}(1,1,0)\text{ with }\mathbf{\hat{p}}=\frac{1}%
{3}\left(  \frac{1}{2}\sqrt{2},\,\sqrt{2}-1,\,1\right)  , \label{S}%
\end{equation}
if%
\begin{equation}
\frac{c_{3}-c_{2}}{u_{1}\alpha}<p_{3,1}-p_{2,1}=\frac{\sqrt{2}}{3}(\sqrt
{2}-1), \label{CS}%
\end{equation}
where the $p_{i,j}$ in (\ref{CF})\ and (\ref{CS}) are evaluated in the
respective steady states. These are usually referred as the
\textit{fundamental} and \textit{speculative} steady states, respectively.

Rearranging the ranking of the storage cost as $c_{3}<c_{2}<c_{1}$ one obtains
steady states similar to those in Model B of KW. The equilibrium
\[
\mathbf{s}_{3}\mathbf{=}(1,0,1)\text{ with }\mathbf{p}=\frac{1}{3}\left(
\sqrt{2}-1,\,1,\,\frac{\sqrt{2}}{2}\right)  \text{,}%
\]
always exists. The equilibrium
\[
\mathbf{s}_{3}\mathbf{=}(0,1,1)\text{ with }\mathbf{p}=\frac{1}{3}\left(
1,\,\frac{\sqrt{2}}{2},\,\sqrt{2}-1\right)  \text{,}%
\]
exists if
\begin{equation}
\frac{c_{3}-c_{1}}{u_{2}\alpha}>p_{3,2}\ -p_{1,2}=\frac{\sqrt{2}}{3}%
(1-\sqrt{2})\ \text{,} \label{cond1b}%
\end{equation}
and%
\begin{equation}
\frac{c_{2}-c_{3}}{u_{1}\alpha}<p_{2,1}=\frac{\sqrt{2}}{6}\text{ }
\label{cond2b}%
\end{equation}
are satisfied, where $p_{1,2}$, $p_{3,2}$, and $p_{1,2}$ are evaluated on the
$\mathbf{s}_{3}\mathbf{=}(0,1,1)$ steady state. Other equilibria emerge under
other rearrangements of the storage costs. Next proposition states which
steady state equilibria are globally stable.

\begin{proposition}
\label{pro:stab-no-money} With the possible exception of $\mathbf{s}%
_{3}=(1,1,1)$, under any other constant set of strategies\ $\mathbf{s}%
_{3}=(s_{2,3}^{1},s_{3,1}^{2},s_{1,2}^{3})$, $\mathbf{\hat{p}}(t)$ converges
to a stationary distribution, $\mathbf{\hat{p}}^{\ast}$, from any
${\mathbf{\hat{p}}}(0)$.
\end{proposition}

\begin{proof}
See Appendix A.
\end{proof}

Consider now the full-fledged model with $M>0$. Adding fiat money into the
model greatly increases the number of steady states that could qualify to be
Nash equilibria. Considering that each type has six possible choices, there
are $6^{3}$ steady states to be verified. Fig. \ref{onethird} illustrates how
variations of $M$ and $\delta_{m}$ affect the emergence of a particular
equilibrium for a Model A economy ($c_{1}<c_{2}<c_{3}$). It considers
equilibria with $\mathbf{s}_{3}=(0,1,0)$ or $\mathbf{s}_{3}=(1,1,0)$ -- the
same two type of equilibria reviewed above for the economy without fiat money.
In reviewing the monetary equilibria, it is useful to keep in mind that money
is accepted for liquidity reasons and to save on storage costs. Money holders,
however, incur a seignorage tax. When this becomes sufficiently large, some
may prefer to face longer waiting times in getting the preferred consumption
good (lower liquidity), and to pay a higher storage cost, rather than holding
fiat money. Since liquidity depends both on the distribution of commodities
and on the magnitude of storage costs, it is conceivable that some do not
accept money even when others do so.

Fig. \ref{onethird} shows that an $\mathbf{s}_{3}=(1,1,0)$ full monetary
equilibrium -- i.e. with $\mathbf{s}_{1}=\mathbf{s}_{2}=(1,1,1)$ -- emerges
for a relatively low stock of fiat money and for low rates of seignorage. As
seignorage gains in importance, however, fiat money becomes less desirable, to
the point that a type 2 is no longer willing to sell good $1$ against fiat
money -- that is, $s_{1,m}^{2}$ switches from $1$ to $0$, so that the set of
strategies becomes $\mathbf{s}=%
\begin{pmatrix}
1 & 1 & 1\\
1 & 0 & 1\\
1 & 1 & 0
\end{pmatrix}
$.

At higher levels of fiat money, good 3 loses its role of commodity money, even
at low seignorage rates, because a type 1's odds of meeting type 3 holding
good 1 shrink. Therefore, a type 1 no longer finds it convenient to pay the
high storage cost of good 3. Hence, the Nash equilibrium is characterized by
$\mathbf{s}=%
\begin{pmatrix}
1 & 1 & 1\\
1 & 1 & 1\\
0 & 1 & 0
\end{pmatrix}
$. At intermediate ranges of the rate of seignorage, $s_{1,m}^{2}$ can be
either $1$ or $0$ or both, that\ is, $\mathbf{s}_{3}=(1,1,0)$ and
$\mathbf{s}_{3}=(0,1,1)$ -- a case of multiple equilibria due to
inflation.\bigskip

\section{Finding Nash Equilibria}

\label{dyna}

This section studies the conditions for obtaining a Nash equilibrium
$({\mathbf{p}}(t),{\mathbf{s}}(t))$ that converges to a steady state one,
starting from an arbitrary initial distribution ${\mathbf{p}}(0)$. It begins
with a proposition that deals with convergence in a neighborhood of a Nash
steady state equilibrium.

\begin{proposition}
\label{closeup} Let $({\mathbf{p}}^{\ast},{\mathbf{s}}^{\ast})$ be a Nash
steady state equilibrium, with ${\mathbf{p}}^{\ast}$ being asymptotically
stable for \eqref{pi}-\eqref{pizero}. There exists an $\epsilon>0$ such that,
if $\Vert{\mathbf{p}}_{0}-{\mathbf{p}}^{\ast}\Vert\leq\epsilon$, the pattern
$({\mathbf{p}}(t),{\mathbf{s}}^{\ast})$, with ${\mathbf{p}}(0)={\mathbf{p}%
}_{0}$, is a Nash Equilibrium.\footnote{We define $\Vert{\mathbf{p}}%
\Vert=\sqrt{\sum_{i,j}p_{i,j}^{2}}$.}
\end{proposition}

\begin{proof}
See Appendix A
\end{proof}

Clearly, when the initial condition ${\mathbf{p}}(0)$ is outside of a small
neighborhood of the Nash steady state, $({\mathbf{p}}^{\ast},{\mathbf{s}%
}^{\ast})$, the Nash set of strategies ${\mathbf{s}}(t)$ may be different than
${\mathbf{s}}^{\ast}$. In such a case, to evaluate whether an ${\mathbf{s}%
}(t)$ is a Nash solution, one needs to verify whether any agent has an
incentive to deviate from such an ${\mathbf{s}}(t)$. In terms of \eqref{sup},
one needs to check if there is any gap between $\mathbf{s}^{i}(t)$ and
${\boldsymbol{\sigma}}^{i}(t)$. The best response ${\boldsymbol{\sigma}}%
^{i}(t)$ can be computed on the basis of the value value functions
$V_{i,j}(t)$ defined in \eqref{sup}. Eq. \eqref{vidot} specifies the time
variation of $V_{i,j}(t)$. Unfortunately the initial value $V_{i,j}(0)$ is
unknown. What is known, however, is the value of $V_{i,j}$ on the Nash steady
state $({\mathbf{p}}^{\ast},{\mathbf{s}}^{\ast})$. This information can be
used to integrate \eqref{vidot} backward in time, starting from a neighborhood
of the Nash steady state. We then proceed as follows.

First, we integrate the system \eqref{pi}-\eqref{pizero} starting from the
initial condition ${\mathbf{p}}(0),$ under the set of strategies
$\mathbf{s}(t)$, an operation that yields ${\mathbf{p}}(t)$, for $0\leq t\leq
T$, where $T$ must be sufficiently large so that $\Vert{\mathbf{p}%
}(T)-\mathbf{p}^{\ast}\Vert\leq\epsilon$. Let $T_{\epsilon}$ be a $T$ that
satisfies this constraint.

Second, we set ${\mathbf{V}}_{i}(T_{\epsilon})={\mathbf{V}}_{i}^{\ast}$, where
${\mathbf{V}}_{i}(t)=\{V_{i,j}(t)\}_{j=i+1,i+2,m}$, and use \eqref{vidot} to
compute ${\mathbf{V}}_{i}(t)$ for $t<T_{\epsilon}$. Said it differently, we
use ${\mathbf{V}}_{i}^{\ast}$ as final condition and integrate \eqref{vidot}
backward in time. The resulting ${\mathbf{V}}_{i}(t)$, however, is not yet
necessarily the true value function because $T_{\epsilon}$ is finite.

Additional details must still be verified to arrive to the true value
function. Let ${\mathbf{V}}_{i}(t,T,\mathbf{W})$ be the unique solution of
\eqref{vidot}, for a given $({\mathbf{s}}(t),{\mathbf{p}}(t))$, with $t\leq T$
under the condition that for $t=T$, ${\mathbf{V}}_{i}(T)={\mathbf{W}}$, i.e.
${\mathbf{V}}_{i}(T,T,{\mathbf{W}})={\mathbf{W}}$. The below proposition
states that it is reasonable to pick a ${\mathbf{W}}={\mathbf{V}}_{i}^{\ast}$
as initial condition to integrate the system \eqref{vidot} backward in time.

\begin{proposition}
\label{compu} Consider a pattern $({\mathbf{s}}(t),{\mathbf{p}}(t))$ that
converges to a steady state Nash equilibrium $({\mathbf{s}}^{\ast}%
,{\mathbf{p}}^{\ast})$. Let ${\mathbf{V}}_{i}^{\ast}$ be the value function
evaluated at the Nash steady state $({\mathbf{s}}^{\ast},{\mathbf{p}}^{\ast}%
)$, and let ${\mathbf{V}}_{i}(t,T,{\mathbf{W}})$ be the unique solution of
\eqref{vidot} with ${\mathbf{V}}_{i}(T,T,{\mathbf{W}})={\mathbf{W}}$. For
every $t\leq T$
\[
\Vert{\mathbf{V}}_{i}(t)-{\mathbf{V}}_{i}(t,T,{\mathbf{V}}_{i}^{\ast}%
)\Vert\leq\sqrt{3}e^{-\rho(T-t)}\Vert{\mathbf{V}}_{i}(T)-{\mathbf{V}}%
_{i}^{\ast}\Vert\,.
\]

\end{proposition}

\begin{proof}
See Appendix A.
\end{proof}

The fact that $\Vert{\mathbf{p}}(T_{\epsilon})-\mathbf{p}^{\ast}\Vert
\leq\epsilon$ implies that $\Vert{\mathbf{V}}_{i}(T_{\epsilon})-{\mathbf{V}%
}_{i}^{\ast}\Vert=O(\epsilon)$. Proposition \ref{compu} states that choosing
${\mathbf{V}}_{i}^{\ast}$ as a boundary value, \eqref{vidot} yields a good
approximation ${\mathbf{V}}_{i}(t,T_{\epsilon},{\mathbf{V}}_{i}^{\ast})$ for
${\mathbf{V}}_{i}(t)$. The response ${\boldsymbol{\sigma}}^{i,\epsilon}(t)$
derived from ${\mathbf{V}}_{i}(t,T_{\epsilon},V_{i}^{\ast})$ is then an
approximation of the best response ${\boldsymbol{\sigma}}^{i}$. In particular,
${\boldsymbol{\sigma}}^{i,\epsilon}$ is a piecewise constant function with a
finite number of switching times $t_{h}^{\epsilon}$, $h=1,\ldots,H$, where one
of the $\sigma_{j,k}^{i,\epsilon}$ changes from 0 to 1 or from 1 to 0. As
$\epsilon$ approaches zero, ${\boldsymbol{\sigma}}^{i,\epsilon}$ converges to
${\boldsymbol{\sigma}}^{i}$ in the sense that $\lim_{\epsilon\rightarrow
0}t_{h}^{\epsilon}=t_{h}$, where $t_{h}$, for $h=1,\ldots,H$, is the finite
set of witching times of ${\boldsymbol{\sigma}}^{i}$. In short, we have
\begin{equation}
{\boldsymbol{\sigma}}^{i}=\boldsymbol{\mathcal{B}}({\mathbf{s}})=\lim
_{\epsilon\rightarrow0}{\boldsymbol{\sigma}}^{i,\epsilon}\text{.}
\label{limit}%
\end{equation}

We deal with the problem of finding a fixed point for the map
$\boldsymbol{\mathcal{B}}$ (see \eqref{fixed}) by designing a simple iterative
scheme. It is convenient to define ${\mathbf{V}}(t)=({\mathbf{V}}_{1}%
(t)$,${\mathbf{V}}_{2}(t)$,${\mathbf{V}}_{3}(t))$. The iteration starts with a
guess ${\mathbf{s}}_{0}(t)$ and then determines the best response
${\boldsymbol{\sigma}}_{0}(t)$ to $\mathbf{s}_{0}(t)$, where
${\boldsymbol{\sigma}}_{0}=\boldsymbol{\mathcal{B}}({\mathbf{s}}_{0})$ and the
associated pattern ${\mathbf{p}}_{0}(t)$. If ${\boldsymbol{\sigma}}%
_{0}(t)={\mathbf{s}}_{0}(t)$, the iteration stops and (${\mathbf{s}}%
_{0}(t),{\mathbf{p}}_{0}(t)$) is the Nash equilibrium. Otherwise, the
iteration continues: It sets a new guess $\mathbf{s}_{1}%
(t)={\boldsymbol{\sigma}}_{0}(t)$ and calculates a new path (${\mathbf{s}}%
_{1}(t),{\mathbf{p}}_{1}(t)$). In general, the iteration generates a sequence
${\mathbf{\sigma}}_{n}=\boldsymbol{\mathcal{B}}({\mathbf{s}}_{n})$. If the
sequence ${\mathbf{s}}_{n}(t)$ converges to a ${\mathbf{s}(t)}$ then the
couple (${\mathbf{s}}(t),{\mathbf{p}}(t)$) is a Nash equilibrium.

More specifically, to find a Nash equilibrium that starts from $\mathbf{p}(0)$
and converges to a Nash steady state (${\mathbf{s}}^{\ast},{\mathbf{p}}^{\ast
}$) we:

\begin{enumerate}
\item Determine the Nash steady state (${\mathbf{s}}^{\ast},{\mathbf{p}}%
^{\ast}$);

\item Chose an initial guess ${\mathbf{s}}_{0}(t)$ for ${\mathbf{s}}(t)$ (for
instance, ${\mathbf{s}}_{0}(t)={\mathbf{s}}^{\ast}$);

\item Integrate \eqref{pi}-\eqref{pizero} forward in time, with ${\mathbf{s}%
}(t)={\mathbf{s}}_{0}(t)$, until the solution ${\mathbf{p}_{0}}(t)$ is
sufficiently close to the steady state ${\mathbf{p}}^{\ast}$;

\item Compute ${\mathbf{V}}(t)$, by integrating \eqref{vidot} backward in time
(with ${\mathbf{V}}^{\ast}$ as final condition) and contemporaneously
determining ${\boldsymbol{\sigma}}_{0}(t)$ through \eqref{consiste}, assuming
${\mathbf{p}}(t)={\mathbf{p}_{0}}(t)$ and ${\mathbf{s}}(t)={\mathbf{s}}%
_{0}(t)$;

\item Set ${\mathbf{s}}_{1}(t)={\boldsymbol{\sigma}}_{0}(t)$ as the new guess
for ${\mathbf{s}}(t)$ and compute a new pattern $({\mathbf{s}}_{1}%
(t),{\mathbf{p}}_{1}(t))$ and a new best response ${\boldsymbol{\sigma}}%
_{1}(t)$;

\item Repeat steps 3 through 5 to generate a sequence (${\boldsymbol{\sigma}%
}_{n}(t)$, ${\mathbf{s}}_{n}(t),{\mathbf{p}}_{n}(t)$);

\item Stop the procedure when the difference between ${\boldsymbol{\sigma}%
}_{n}(t)$ and ${\mathbf{s}}_{n}(t)$ is smaller than a predetermined error
$\epsilon$.\footnote{We define the distance between ${\boldsymbol{\sigma}}%
_{n}$ and ${\mathbf{s}}_{n}$ as the $\max_{i=1,\ldots,K}|t_{i}-\tau_{i}|$,
where $t_{i}$ is the switching time in ${\boldsymbol{\sigma}}_{n}$ and
$\tau_{i}$ that in ${\mathbf{s}}_{n}$. The ($\mathbf{s}(t),{\mathbf{p}}(t)$)
obtained by taking $\epsilon\rightarrow0$ is a Nash equilibrium.}
\end{enumerate}

Two observations are in order. First, there are no issues of instability when
computing ${\mathbf{p}}(t)$\ and ${\mathbf{V}}(t)$: The system
\eqref{pi}-\eqref{pizero} is stable when integrated forward in time and so is
\eqref{vidot} if integrated backward in time. The iteration does not need to
converge but if it does, it necessarily converges to a fixed point of
$\boldsymbol{\mathcal{B}}$. In case of non-convergence, more refined iterative
schemes or a Newton-Raphson method could be employed. Nevertheless, in all our
numerical experiments (Section 5) this simple iterative algorithm delivered a
fixed point.

Second, the design of the algorithm presents similarities with the stable
manifold theorem for ordinary differential equations (see Appendix B for a
formal discussion). It differs, however, from the standard approaches employed
to study transitional dynamics of macroeconomic and growth models. Indeed, it
is common to compute an equilibrium pattern by integrating a system of
differential equations that describe the equilibrium conditions of the economy
backward in time, starting from a neighborhood of the steady state (see
Brunner and Strulik, 2002). This approach usually works well when the
dimension of the system is small. At high dimensions it is sometimes possible
to approximate the manifold in a neighborhood of the steady state by means of
projection methods (see McGrattan, 1999, and Mulligan and Sala-i-Martin,
1993). Nevertheless, when the dimension of the system is large, constructing
the manifold in regions away from the steady state is generally very
problematic -- a serious limitation when the choice of an initial condition
away from the steady state is an important aspect of the exercise.


\section{Numerical Experiments}

\label{numeri} This section proposes a few applications of the dynamic
analysis. First, it illustrates the transition from partial to full acceptance
of fiat money. Second, it studies the effects of changes in the rate of
seignorage. It then discusses issues of multiple equilibria related to
seignorage and to the distribution of the population across the three types.
Finally, it briefly reviews equilibria in Model B.

\subsection{Partial and Full Acceptance of Fiat Money}

Because the conditions for the full acceptance of money in a steady state are
different than those in other regions of the inventory space, an economy may
go, while converging to a full monetary equilibrium, through a phase in which
some do not accept fiat money. For one, along the transition the degree of
acceptability of a low-storage commodity may simply decline, and thus favors
the acceptability of fiat money. In addition, the cost of seignorage, given by
$V_{i,m}-V_{i,i+1}-D_{i}$, may also decline along the transition -- the
production cost, $D_{i}$, is constant over time but the difference
$V_{i,m}-V_{i,i+1}$ is not. It could be, for instance, that along the
transition the liquidity of commodity $i+1$ drops, implying a reduction of the
cost of seignorage. Panel A of fig. \ref{double_emergence} illustrates such a
scenario in the phase diagram. Initially, $\mathbf{s}=%
\begin{pmatrix}
0 & 1 & 1\\
1 & 0 & 1\\
0 & 1 & 0
\end{pmatrix}
$, that is, type 1 prefers good 2 to fiat money, and type 2 prefers good 1 to
fiat money. The economy eventually converges to a full monetary equilibrium
$\mathbf{s}=%
\begin{pmatrix}
1 & 1 & 1\\
1 & 1 & 1\\
1 & 1 & 0
\end{pmatrix}
$. Observe that along the transition also good 3 acquires the role of
commodity money, as type 1 agents switch from fundamental to speculative
strategies. Higher rates of seignorage delay the emergence of money (see
Panels B and C of fig. \ref{double_emergence}).

\subsection{A Monetary Reform}

It is a long-standing tenet of economics that inflation may create
inefficiencies because it distorts the choices of individuals. To understand
how people's behavior is affected by seignorage, consider an economy currently
on full monetary steady state equilibrium with $\mathbf{s}_{3}=(0,1,0)$.
Through a reduction of the seignorage rate, the government may pull the
economy out this equilibrium and send it to an $\mathbf{s}_{3}=(1,1,0)$ full
monetary equilibrium. In the latter equilibrium agents on average trade more
frequently and produce at a faster rate than in the former one. A $2$
percentage points reduction of $\delta_{m}$, from the initial state $M=0.3$
and $\delta_{m}=0.1$, would be sufficient to accomplish the task (see fig.
\ref{onethird}). The dynamic consequences of the shock are depicted in the
phase diagram of fig. \ref{reduction_seignorage}. Because the drop in the
seignorage rate increases the value of fiat money relative to that of all
other commodities, type 2 agents are induced sell good 1 against fiat money.
In addition, because the gap $p_{2,1}-p_{3,1}$ declines along the adjustment
to the new equilibrium, type 1 agents, in anticipation of such liquidity
change, immediately switch from fundamental to speculative strategies
($\mathbf{s}_{3}$ turns from $(0,1,0)$ to $(1,1,0)$). As a result, production
booms right after the shock and then stabilizes at a higher level relative to
that of the initial equilibrium (see fig.~\ref{reduction_seignorage}b).

\subsection{Inflation and Beliefs}

It is often argued that the effects of inflation on production depends on the
coordination of beliefs. This conjecture, in our framework, emerges in fig.
\ref{onethird} that reviews the type of equilibria associated with different
combinations of $M$ and $\delta_{m}$. The figure shows that in some regions of
the $(M,\delta_{m})$ space, two steady state equilibria exist. For instance,
when $M=0.3$, and $\delta_{m}$ is between 5 and 9 percent, the full monetary
equilibrium $\mathbf{s}=%
\begin{pmatrix}
1 & 1 & 1\\
1 & 1 & 1\\
1 & 1 & 0
\end{pmatrix}
$ coexists with the equilibrium $\mathbf{s}=%
\begin{pmatrix}
1 & 1 & 1\\
1 & 0 & 1\\
0 & 1 & 0
\end{pmatrix}
$.

This means that if initially the economy is at a (unique) steady state
equilibrium $\mathbf{s}=%
\begin{pmatrix}
1 & 1 & 1\\
1 & 0 & 1\\
0 & 1 & 0
\end{pmatrix}
$, with $M=0.3$ and $\delta_{m}=0.1$, a reduction of the seignorage rate, from
10 to, for example, 6 percent, may or may not induce type 1 agents to switch
from fundamental to speculative strategies ($\mathbf{s}_{3}$ may or may not
change from $(0,1,0)$ to $(1,1,0)$). Agents may keep their actions coordinated
on the current $\mathbf{s}_{3}=(0,1,0)$ equilibrium, in which case the policy
intervention generates only marginal changes in the economy. Conversely, they
may coordinate their actions on the $\mathbf{s}_{3}=(1,1,0)$ equilibrium -- in
which case the adjustment process would be very similar to that depicted in
fig. \ref{reduction_seignorage}.

In brief, a modest reduction of the seignorage rate associated with a dose of
optimism can be effective in stirring up production. But if agents are
unresponsive to relatively modest changes into the seignorage rate, the
government would need to implement a more radical monetary reform, and be
prepared to give up a larger share of its current seignorage revenue -- in our
example, below $\delta_{m}=0.05$ there is a unique equilibrium with
$\mathbf{s}_{3}=(1,1,0)$.

\subsection{Uneven Distribution of Types and Multiple Equilibria}

The distribution of the population across types affects the emergence of a
particular equilibrium. Wright (1995)\ established the existence of multiple
equilibria in a Model A economy without fiat money. In particular, when the
share of type 3 agents is relatively high, the speculative equilibrium
$\mathbf{s}_{3}=(1,1,0)$ coexists with one in which all three agents flip
their strategies, that is a $\mathbf{s}_{3}=(0,0,1)$ equilibrium. Seignorage
promotes the emergence of additional equilibria. Fig.
\ref{monetary_equilibria}a shows how the values of $\theta_{i}$ conditions the
emergence of a particular equilibrium, for a given stock of real balances and
of the seignorage rate. In the symmetric distribution $\theta_{i}=\frac{1}{3}%
$, under the current specification ($M=0.3$, $\delta_{m}=0.02$; see also table
\ref{parameters}), a unique full monetary $\mathbf{s}_{3}=$(1,1,0) equilibrium
exists. When the share of type 3 agents is slashed by half, that is,
$\theta_{3}=\frac{1}{6}$, $\theta_{2}=\frac{1}{3}$, and $\theta_{1}=\frac
{1}{2}$, the equilibrium is characterized by $\mathbf{s}_{3}=$(1,1,0) and
still full acceptance of money. Conversely, a reduction of the share type 2
agents means that the original full monetary $\mathbf{s}=%
\begin{pmatrix}
1 & 1 & 1\\
1 & 1 & 1\\
1 & 1 & 0
\end{pmatrix}
$ equilibrium is coupled with $\mathbf{s}=%
\begin{pmatrix}
1 & 1 & 1\\
1 & 0 & 1\\
1 & 1 & 0
\end{pmatrix}
$. A unique equilibrium with partial acceptability of fiat money is observed
when $\theta_{1}$ is high and $\theta_{3}$ is low -- a situation in which type
2 easily trades 1 for 2.

As higher seignorage rates, regions supported by a partial acceptability
$\mathbf{s}_{2}=(1,0,1)$ expand. For instance, a comparison of fig.
\ref{monetary_equilibria}a and fig. \ref{monetary_equilibria}b reveals that
when seignorage rate goes from $2$ to $10$ per cent, the full monetary
fundamental equilibrium disappears, and the overlapping between the equilibria
$\mathbf{s}=%
\begin{pmatrix}
1 & 1 & 1\\
1 & 0 & 1\\
1 & 1 & 0
\end{pmatrix}
$ and $\mathbf{s}=%
\begin{pmatrix}
1 & 1 & 1\\
1 & 1 & 1\\
1 & 1 & 0
\end{pmatrix}
$ is more commonly observed.\medskip

\noindent\textit{Multiple Equilibria}. The presence of multiple monetary
steady states does not necessarily imply the existence of multiple Nash
equilibria. In principle, it could be that once the initial condition is
specified, the pattern converges to one and only one steady state. Fig.
\ref{multiple_equilibria} clarifies, however, that in our environment there is
multiplicity: Two economies with the same set of parameters and the same
initial condition coordinate on different steady state equilibria.

\subsection{Uneven Distribution of Types and Dynamics in Model B}

This section briefly discusses the acceptability of fiat money in a Model B
economy, for different values of $\theta_{i}$. The effects of a rise in the
seignorage rate can be learned by comparing the equilibria in fig.
\ref{overview_equilibria_B}b where $\delta_{m}=0.1$ and those in fig.
\ref{overview_equilibria_B}a where $\delta_{m}=0.02$. In the low-inflation
economy larger regions of full monetary equilibria overlaps with similar
equilibria in which $s_{3,m}^{1}=0.$ In the space just below the 45-degree
line, there are only full monetary fundamental equilibria $\mathbf{s}%
_{3}=(1,0,1)$ because the scarcity of type 1 agents reduces the liquidity
value of good 3 -- there are too few middle-men that bring good 3 from type 2
to type 3 agents. Indeed, good 3 is dominated by fiat money even with a
seignorage rate three times larger than $c_{3}$. For a more balanced
distribution of the population, however, at a sufficiently high seignorage
rates, equilibria with partial acceptability of money ($s_{3,m}^{1}=0$) become unique.

As with model A, the liquidity conditions and the cost of seignorage change
along the dynamics, implying that for a given specification of the economy,
the fraction of the population that accepts money changes along the transition
to the Nash steady state equilibrium. Fig. \ref{emergencemodelB} shows a
particular scenario in which as the economy converges to the $\mathbf{s}%
_{3}=(1,0,1)$ full monetary steady state equilibrium, type 3 agents switch
their strategies with respect to fiat money when holding good 2. Along the
transition both good 2 and fiat money become more valuable (see middle plot of
fig. \ref{emergencemodelB}), as their liquidity improves, but the value of
fiat money increases more rapidly and eventually catches up with the value
good 2.

\section{Welfare}

One standard question of monetary economics is whether the acceptance of fiat
money improves the allocation of resources and stimulates production. The
presence of matching frictions and the assumption that agents incur a cost in
holding commodities gives fiat money a potential positive role. Nevertheless,
it also comes with costs at the individual's and society's level. At the
individual level, it looms the risk of confiscation. At the society level,
fiat money reduces the availability of commodities -- money chases away
consumption goods. To explores the welfare implications of introducing fiat
money and of altering seignorage rates, we use, as KW, a utilitarian welfare
criterion -- given the highly symmetric type of environment it is unlikely to
observe a Pareto improvement from any given state. The payoff of a type $i$
agent is calculated as a weighted average of $V_{i,j}$:
\[
W_{i}(t)=\frac{1}{\theta_{i}}(p_{i,i+1}V_{i,i+1}+p_{i,i+2}V_{i,i+2}%
+p_{i,m}V_{i,m}).
\]
Therefore, the welfare of the whole society is simply the average of the three
groups' payoffs:
\[
W(t)=\sum_{i}\theta_{i}W_{i}(t)\text{.}%
\]

As mentioned in the introduction the government derives utility from consuming
goods. These are purchased through the seignorage tax $\delta_{m}M$. One may
argue that the government cares also about the welfare of the population. This
could reflect a genuine interest in the society's well-being, or more simply
the desire of maintaining the population's electoral support. The government
welfare function is then
\[
W^{G}(t)=(1-\lambda)Q(t)+\lambda W(t),
\]
where
\[
Q(t)=M\int_{t}^{+\infty}e^{-\rho^{G}(\tau-t)}\delta_{m}(\tau)d\tau,
\]
and where $\lambda$ weighs the society's welfare in the government's objective
function. For the sake of the illustration, we consider an altruistic
government ($\lambda=1$) that wants to know how the level of $M$ or the rate
of seignorage affects the population. \medskip

\noindent\textit{Real Balances.} Fig. \ref{welfareM}a shows the welfare levels
on steady state for different levels of fiat money and zero seignorage. When
the stock of fiat money is relatively large, any further increase tends to
make people, on average, worse-off because the chase-away-good effect largely
dominates. Observe that a change in $M$ can induce $W_{1},$ $W_{2}$ and
$W_{3}$ to move in different directions, indicating a conflict of interest in
the society about the desirable level of real balances. Type 3 individuals,
for instance, prefer a lower level of fiat money than the other two groups. As
they carry the low storage cost good, relative to the other two groups, they
are more preoccupied by the displacement of commodities caused by a further
increase in fiat money than pleased by the savings in storage costs. \medskip

\noindent\textit{Seigniorage.} When the government confiscates fiat money, it
clearly reduces the welfare of the targeted individuals by $V_{i,m}%
-V_{i,i+1}-D_{i}$. But it also alters the odds that an individual is able to
exchange his commodity against money. While a private agent in a match may
refuse to buy, when a government agent carrying money meets a private agent,
there is always an exchange. Hence, seignorage can act as a stimulus to
production. But with seignorage comes also some income redistribution. First,
as noted earlier, the cost of seignorage differs across types: Because a type
3 who relinquishes fiat money produces the good with the lowest storage cost,
his burden of seignorage is lighter than that of the other two types. Second,
because real balances are in general not held in equal proportions across the
three groups, the probability of being hit by seignorage differs across types.
Similarly, the probability that the government purchases a commodity depends
also on how commodities are distributed across types. Fig. \ref{welfareM}b,
for instance, shows that over a certain range of seignorage, an increase in
seignorage causes a decline in $W_{1}$ and $W_{2}$ and an increase in $W_{3}$.
Interestingly, the figure also says that an increase in seignorage can boost
everybody's welfare $W_{i}$ when starting from low levels of seignorage.

\section{Further Research}

The set up of the problem (Section 2) and the procedure to find Nash
equilibria (Section 4) are valid for a more general search model with $N$
goods, and $N$ types of agents, as, for instance, in Aiyagari and Wallace
(1991, 1992) (see Appendix A). The analysis could also be adapted to allow for
multiple holdings, as in Molico (2006), Lagos and Rocheteau (2008), and Chiu
and Molico (2010), and to study the dynamics of indivisible-asset models in
which heterogeneity is an essential ingredient, such as studies of the
middlemen by Rubinstein and Wolinsky (1987), international currency by
Matsuyama et al. (1993), banking by Cavalcanti and Wallace (1999), and
over-the-counter financial markets by Duffie et al. (2005).

\bigskip\noindent{\LARGE References}\bigskip

\noindent Aiyagari, S. and N. Wallace (1991) Existence of steady states with
positive consumption in the Kiyotaki-Wright model. \textit{Review of Economic
Studies}, 58, 901-16.\medskip

\noindent Aiyagari, S. and N. Wallace (1992) Fiat money in the Kiyotaki-Wright
model. \textit{Economic Theory} 2, 447-64.\medskip

\noindent Araujo, L. (2004) Social norms and money. \textit{Journal of
Monetary Economics}, 51, 241-256.\medskip

\noindent Araujo, L., Camargo B., Minetti R, and Puzzello D. (2012) The
essentiality of money in environment with centralized trade. \textit{Journal
of Monetary Economics}, 59, 612-621.\medskip

\noindent Ba\c{s}\c{c}i, E. (1999) Learning by imitation. \textit{Journal of
Economic Dynamics and Control}, 23(9-10), 1569-1585.\medskip

\noindent Baumol, W. (1952) The transactions demand for cash: An inventory
theoretic approach. \textit{The Quarterly Journal of Economics}, 66(4),
545-556.\medskip

\noindent Brown, P. M. (1996) Experimental evidence on money as a medium of
exchange. \textit{Journal of Economic Dynamics and Control}, 20(4),
583-600.\medskip

\noindent Brunner, M. and H. Strulik (2002) Solution of perfect foresight
saddlepoint problems: A simple method and applications. \textit{Journal of
Economic Dynamics and Control}, 26(5), 737-53.\medskip

\noindent Cavalcanti, R. O. and N. Wallace (1999) Inside and outside money as
alternative media of exchange. \textit{Journal of Money, Credit and Banking},
31 (3), 443-457.\medskip

\noindent Chiu, J. and M. Molico (2010) Liquidity, redistribution, and the
welfare cost of inflation. \textit{Journal of Monetary Economics}, 57(4),
428-38.\medskip

\noindent Deviatov, A. and Wallace, N. (2014) Optimal inflation in a model of
inside money. \textit{Review of Economic Dynamics}, 17, 287--293.\medskip

\noindent Diamond, P. (1982) Aggregate demand management in search
equilibrium. \textit{Journal of Political Economy}, 90, 881-94.\medskip

\noindent Duffie, D., N. G\^{a}rleanu and L. Pederson (2005) Over-the-counter
markets.\textit{\ Econometrica} 73, 1815-1847.\medskip

\noindent Duffy, J., and Jack Ochs (1999) Emergence of money as a medium of
exchange: An experimental study. \textit{American Economic Review}, 89 (4),
847-877.\medskip

\noindent Duffy, J., and Jack Ochs (2002) Intrinsically worthless objects as
media of exchange: Experimental evidence. \textit{International Economic
Review}, 43 (3), 637-674.\medskip

\noindent Goetzmann, W. N. (2016) \textit{Money changes everything: How
finance made civilization possible}. Princeton University Press.\medskip

\noindent Kehoe, M. J., N. Kiyotaki, and R. Wright (1993) More on money as a
medium of exchange. \textit{Economic Theory} 3, 297-314.\medskip

\noindent Kiyotaki, N., and R. Wright (1989) On money as a medium of exchange.
\textit{Journal of Political Economy} 97, 927--954.\medskip

\noindent Iacopetta, M. (2018) The emergence of money: A\ dynamic analysis.
\textit{Macroeconomic Dynamics, }forthcoming.\medskip

\noindent Jovanovic, B. (1982). Inflation and welfare in the steady state.
\textit{Journal of Political Economy}, 90, 561-77. \medskip

\noindent Lagos, R. and G. Rocheteau, (2009), Liquidity in markets with search
frictions.\ \textit{Econometrica}, 77(2), 403-26.\medskip

\noindent Lagos, R., G. Rocheteau, and R. Wright (2017) Liquidity: A new
monetarist perspective. \textit{Journal of Economic Literature} 55,
371-440.\medskip

\noindent Li, V., (1994) Inventory accumulation in a search-based monetary
economy. \textit{Journal of Monetary Economics}, 34, 511-36.\medskip

\noindent Li, V., (1995) The optimal taxation of fiat money in search
equilibrium. \textit{International Economic Review}, 36(4), 927-942.\medskip

\noindent Luo, G. (1999) The evolution of money as a medium of
exchange.\ \textit{Journal of Economic Dynamics and Control}, 23,
415-58.\medskip

\noindent Mankiw, N. G., (2006) \textit{Principles of Economics}, 4th ed.,
Cengage Learning.\medskip

\noindent Marimon, R., E. McGrattan, and T. Sargent (1990) Money as a medium
of exchange in an economy with artificially intelligent
agents.\ \textit{Journal of Economic Dynamics and Control,} 14,
329-73.\medskip

\noindent Matsuyama, K., N., Kiyotaki, and A. Matsui (1993) Toward a theory of
international currency. \textit{Review of Economic Studies}, 60(2),
283-307.\medskip

\noindent McGrattan, E.R., (1999) Application of weighted residual methods to
dynamic economic models. In: Marimon, R., Scott, A. (Eds.),
\textit{Computational Methods For The Study Of Dynamic Economies}. Oxford
University Press, Oxford, 114-142.\medskip

\noindent Molico, M. (2006) The distribution of money and prices in search
equilibrium. \textit{International Economic Review}, 47(3), 701-22.\medskip

\noindent Mulligan, C. B. and X. Sala-i-Martin (1993) Transitional dynamics in
two-sector models of endogenous growth, \textit{The Quarterly Journal of
Economics}, 108(3), 739--773.\medskip

\noindent Nash, J. (1950) Equilibrium points in $n$-person games.
\textit{Proceedings of the National Academy of Sciences} 36(1): 48-49.\medskip

\noindent Oberfield, E. and T. Trachter (2012) Commodity money with frequent
search. \textit{Journal of Economic Theory}, 147, 2332-56.\medskip

\noindent Quiggin, A. H. (1949) \textit{A\ survey of primitive money: the
beginning of currency}. Barnes and Noble, New York.\medskip

\noindent Renero, J.M. (1998) Unstable and stable steady-states in the
Kiyotaki-Wright model. \textit{Economic Theory} 11, 275-294.\medskip

\noindent Robinson, C., R., (1995) Dynamical systems, stability, symbolic
dynamics, and chaos.\ \textit{Studies in Advanced Mathematics}, CRC Press,
Boca Raton.\medskip

\noindent Rubinstein, A. and A. Wolinsky (1987) Middlemen. \textit{Quarterly
Journal of Economics}, 102, 581-94.\medskip

\noindent Shevchenko, A. and R. Wright (2004) A simple search model of money
with heterogeneous agents and partial acceptability. \textit{Economic Theory},
24(4), 877-885.\medskip

\noindent Sethi, R. (1999) Evolutionary stability and media of
exchange.\textit{\ Journal of Economic Behavior and Organization}, 40,
233-54.\medskip

\noindent Tobin, J (1956) The interest elasticity of transactions demand for
cash. \textit{The Review of Economics and Statistics}, 38(3), 241--47.\medskip

\noindent Wright, R. (1995) Search, evolution, and money. \textit{Journal of
Economic Dynamics and Control}, 19, 181-206.

\newpage

\appendix

\section{Proofs and derivations}

\label{app:proofs}

This Appendix contains the proofs of Propositions 1 to 5. From the proofs of
Propositions 1, 2, 4, and 5 it will emerge that their statements apply to a
more general setting with $N$ objects and $N$ types of agents. Specifically,
the size of the matrix $\mathcal{A}^{i}$, defined in \eqref{AA} in this
Appendix, can be augmented to consider more objects, and the index $i$,
associated with types of individuals, can run up to an $N>3$.

\begin{proof}
[Proof of Proposition \ref{Vidot}]Differentiation of \eqref{Vf} with respect
to $t$ yields
\begin{equation}
\dot{V}_{i,j}(t)=-v_{i,j}({\mathbf{p}}(t))+\rho V_{i,j}(t)+\int_{t}^{\infty
}e^{-\rho(\tau-t)}\sum_{l}\dot{\pi}_{l,j}^{i}(\tau,t)v_{i,l}({\mathbf{p}}%
(\tau))d\tau\text{,} \label{Vfdot}%
\end{equation}
where
\begin{equation}
v_{i,j}({\mathbf{p}})=\sum_{i^{\prime}}p_{i^{\prime},i}s_{j,i}^{i^{\prime}%
}u_{i}-c_{j} \label{vv}%
\end{equation}
(when $j=m$, $c_{m}$ stands for $\delta_{m}D$) is the expected utility from
consumption, net of storage cost, for an agent of type $i$ with good $j$, and
where $\dot{\pi}_{k,j}^{i}(\tau,t)\equiv\frac{d}{dt}\pi_{k,j}^{i}(\tau,t)$. To
derive $\frac{d}{dt}\pi_{k,j}^{i}(\tau,t)$, first observe that expressions
similar to \eqref{pi}-\eqref{pizero} imply that%

\begin{align}
\frac{d}{d\tau}\pi_{i+1,j}^{i}(\tau,t)=  &  \alpha\left\{  \sum_{i^{\prime}%
}\sum_{k}\pi_{k,j}^{i}p_{i^{\prime},i+1}\sigma_{k,i+1}^{i}s_{i+1,k}%
^{i^{\prime}}+\sum_{i^{\prime}}\pi_{k,j}^{i}p_{i^{\prime},i}s_{i,k}%
^{i^{\prime}}-\sum_{i^{\prime}}\sum_{k}\pi_{i+1,j}^{i}p_{i^{\prime},k}%
\sigma_{i+1,k}^{i}s_{k,i+1}^{i^{\prime}}\right\} \nonumber\\
-  &  \delta_{g}\pi_{i+1,j}^{i}+\delta_{m}\pi_{m,j}^{i}\label{ppi}\\
\frac{d}{d\tau}\pi_{i+2,j}^{i}(\tau,t)=  &  \alpha\left\{  \sum_{i^{\prime}%
}\sum_{k}\pi_{k,j}^{i}p_{i^{\prime},i+2}\sigma_{k,i+2}^{i}s_{i+2,k}%
^{i^{\prime}}-\sum_{i^{\prime}}\sum_{k}\pi_{i+2,j}^{i}p_{i^{\prime},k}%
\sigma_{i+2,k}^{i}s_{k,i+2}^{i^{\prime}}\right\}  -\delta_{g}\pi_{i+2,j}%
^{i}\label{ppitwo}\\
\frac{d}{d\tau}\pi_{m,j}^{i}(\tau,t)=  &  \alpha\left\{  \sum_{i^{\prime}}%
\sum_{k}\pi_{k,j}^{i}p_{i^{\prime},m}\sigma_{k,m}^{i}s_{j,m}^{i^{\prime}}%
-\sum_{i^{\prime}}\sum_{k}\pi_{m,j}^{i}p_{i^{\prime},k}\sigma_{m,k}^{i}%
s_{k,m}^{i^{\prime}}\right\} \nonumber\\
-  &  \delta_{m}\pi_{m,j}^{i}+\delta_{g}(\pi_{i+1,j}^{i}+\pi_{i+2,j}^{i})
\label{ppizero}%
\end{align}
with initial condition $\pi_{k,j}^{i}(t,t)=1$ if $k=j$, and 0 otherwise. For
each $i=1,2,3$, we consider the $3\times3$ matrix $\mathcal{A}^{i}=\left\{
\mathcal{A}_{j,k}^{i}\right\}  _{j,k=i+1,i+2,m}$ defined as:
\begin{align}
\mathcal{A}_{i+1,i+1}^{i}=  &  -\alpha\sum_{i^{\prime}}\sum_{k\not =%
i+1}p_{i^{\prime},k}\sigma_{i+1,k}^{i}s_{k,i+1}^{i^{\prime}}-\delta
_{g}\nonumber\\
\mathcal{A}_{i+1,i+2}^{i}=  &  \alpha\sum_{i^{\prime}}p_{i^{\prime},i+2}%
\sigma_{i+1,i+2}^{i}s_{i+2,i+1}^{i^{\prime}}\nonumber\\
\mathcal{A}_{i+1,m}^{i}=  &  \alpha\sum_{i^{\prime}}p_{i^{\prime},m}%
\sigma_{i+1,m}^{i}s_{m,i+1}^{i^{\prime}}+\delta_{g}\nonumber\\
&  \hbox{}\nonumber\\
\mathcal{A}_{i+2,i+1}^{i}=  &  \alpha\sum_{i^{\prime}}p_{i^{\prime},i+1}%
\sigma_{i+2,i+1}^{i}s_{i+1,i+2}^{i^{\prime}}+\alpha\sum_{i^{\prime}%
}p_{i^{\prime},i}s_{i,i+2}^{i^{\prime}}\nonumber\\
\mathcal{A}_{i+2,i+2}^{i}=  &  -\alpha\sum_{i^{\prime}}\sum_{k\not =%
i+2}p_{i^{\prime},k}\sigma_{i+2,k}^{i}s_{k,i+2}^{i^{\prime}}-\alpha
\sum_{i^{\prime}}p_{i^{\prime},i}s_{i,i+2}^{i^{\prime}}-\delta_{g}\label{AA}\\
\mathcal{A}_{i+2,m}^{i}=  &  \alpha\sum_{i^{\prime}}p_{i^{\prime},m}%
\sigma_{i+2,m}^{i}s_{m,i+2}^{i^{\prime}}+\delta_{g}\nonumber\\
&  \hbox{}\nonumber\\
\mathcal{A}_{m,i+1}^{i}=  &  \alpha\sum_{i^{\prime}}p_{i^{\prime},i+1}%
\sigma_{m,i+1}^{i}s_{i+1,m}^{i^{\prime}}+\alpha\sum_{i^{\prime}}p_{i^{\prime
},i}s_{i,m}^{i^{\prime}}+\delta_{m}\nonumber\\
\mathcal{A}_{m,i+2}^{i}=  &  \alpha\sum_{i^{\prime}}p_{i^{\prime},i+2}%
\sigma_{m,i+2}^{i}s_{i+2,m}^{i^{\prime}}\nonumber\\
\mathcal{A}_{m,m}^{i}=  &  -\alpha\sum_{i^{\prime}}\sum_{k\not =m}%
p_{i^{\prime},k}\sigma_{m,k}^{i}s_{k,m}^{i^{\prime}}-\alpha\sum_{i^{\prime}%
}p_{i^{\prime},i}s_{i,m}^{i^{\prime}}-\delta_{m}\,.\nonumber
\end{align}
The expressions in \eqref{ppi}-\eqref{ppizero} then simplify to
\[
\frac{d}{d\tau}\pi_{k,j}^{i}(\tau,t)=\sum_{l}\mathcal{A}_{l,k}^{i}(\tau
)\pi_{l,j}^{i}(\tau,t).
\]
Observe that the matrix ${\mathcal{A}}^{i}$ satisfies ${\mathcal{A}}_{j,k}%
^{i}\geq0$ for $j\not =k$ and
\begin{equation}
\sum_{l}\mathcal{A}_{j,l}^{i}=0\,, \label{PF}%
\end{equation}
for every $j$. These properties are used in proving Proposition \ref{best} and
Lemma \ref{PFP} (see below).

Calling $\Pi^{i}(\tau,t)$ the $3\times3$ matrix with entries $(\Pi^{i}%
)_{k,j}(\tau,t)=\pi_{k,j}^{i}(\tau,t)$, the above equation can be written as
\begin{equation}
\frac{d}{d\tau}\Pi^{i}(\tau,t)=\mathcal{A}^{i}(\tau)^{T}\Pi^{i}(\tau,t)
\label{Pidot}%
\end{equation}
where $\mathcal{A}^{i}(t)^{T}$ is the transpose of $\mathcal{A}^{i}(t)$. What
is needed to compute the evolution of ${V}_{i,j}(t)$, however, is the
derivative of $\Pi^{i}(\tau,t)$ with respect to $t$ rather than with respect
to $\tau$. Note, however, that from \eqref{Pidot} it follows that
\[
\Pi^{i}(\tau,t-dt)=\Pi^{i}(\tau,t)\Pi^{i}(t,t-dt)=\Pi^{i}(\tau,t)\left(
1+dt\mathcal{A}^{i}(t)^{T}\right)  .
\]
Consequently,
\[
\frac{d}{dt}\Pi^{i}(\tau,t)=-\Pi^{i}(\tau,t)\mathcal{A}^{i}(t)^{T},
\]
that in extended form becomes
\begin{equation}
\frac{d}{dt}\pi_{k,j}^{i}(\tau,t)=-\sum_{l}\mathcal{A}_{j,l}^{i}(t)\pi
_{k,l}^{i}(\tau,t)\text{.} \label{ppit}%
\end{equation}
By inserting \eqref{ppit} in \ref{Vfdot} we obtain
\begin{equation}
\dot{\mathbf{V}}_{i}=\rho{\mathbf{V}}_{i}-\mathcal{A}^{i}(t){\mathbf{V}}%
_{i}-{\mathbf{v}}^{i}(t) \label{vvidot}%
\end{equation}
where ${\mathbf{v}}^{i}(t)=(v_{i,i+1}({\mathbf{p}}(t)),v_{i,i+2}({\mathbf{p}%
}(t)),v_{i,m}({\mathbf{p}}(t)))$. Using the definition of the matrix
$\mathcal{A}^{i}$ in \eqref{AA} and of $v_{i,j}$ in \eqref{vv} we get, for
$j=i+1$ or $i+2$,
\begin{align*}
\dot{V}_{i,j}=  &  \rho V_{i,j}-\alpha\left\{  \sum_{i^{\prime}}\sum
_{k\not =i}p_{i^{\prime},k}\sigma_{j,k}^{i}s_{k,j}^{i^{\prime}}V_{i,k}%
+\sum_{i^{\prime}}p_{i^{\prime},i}s_{i,j}^{i^{\prime}}V_{i,i+1}-\sum
_{i^{\prime},k}p_{i^{\prime},k}\sigma_{j,k}^{i}s_{k,j}^{i^{\prime}}%
V_{i,j}\right\} \\
-  &  \delta_{g}(V_{i,m}-V_{i,j})-\sum_{i^{\prime}}p_{i^{\prime},i}%
s_{j,i}^{i^{\prime}}u_{i}+c_{j}\text{.}%
\end{align*}
Rearranging terms we get:%
\begin{align*}
\dot{V}_{i,j}=  &  (\rho+\delta_{g}+\alpha)V_{i,j}-\alpha\left\{
\sum_{i^{\prime}}\sum_{k\not =i}p_{i^{\prime},k}\sigma_{j,k}^{i}%
s_{k,j}^{i^{\prime}}V_{i,k}+\sum_{i^{\prime}}p_{i^{\prime},i}s_{i,j}%
^{i^{\prime}}(V_{i,i+1}+u_{i})+\right. \\
+  &  \left.  \sum_{i^{\prime},k}p_{i^{\prime},k}(1-\sigma_{j,k}^{i}%
s_{k,j}^{i^{\prime}})V_{i,j}\right\}  -\delta_{g}V_{i,m}+c_{j}\text{.}%
\end{align*}
This is the expression in \eqref{vidot} when $j=i+1$ or $j+2$. Similarly, when
$j=m$
\begin{align*}
\dot{V}_{i,m}=  &  \rho V_{i,j}-\alpha\left\{  \sum_{i^{\prime}}\sum
_{k\not =i}p_{i^{\prime},k}\sigma_{m,k}^{i}s_{k,m}^{i^{\prime}}V_{i,k}%
+\sum_{i^{\prime}}p_{i^{\prime},i}s_{i,m}^{i^{\prime}}V_{i,i+1}+\sum
_{i^{\prime},k}\sigma_{m,k}^{i}s_{k,m}^{i^{\prime}}V_{i,j}\right\} \\
-  &  \delta_{m}(V_{i,i+1}-V_{i,m})-\sum_{i^{\prime}}p_{i^{\prime},i}%
s_{i,m}^{i^{\prime}}u_{i}-\delta_{m}D_{i}=\\
=  &  (\rho+\delta_{g}+\alpha)V_{i,m}-\alpha\left\{  \sum_{i^{\prime}}%
\sum_{k\not =i}p_{i^{\prime},k}\sigma_{m,k}^{i}s_{k,m}^{i^{\prime}}%
V_{i,k}+\sum_{i^{\prime}}p_{i^{\prime},i}s_{i,m}^{i^{\prime}}(V_{i,i+1}%
+u_{i})+\right. \\
+  &  \left.  \sum_{i^{\prime},k}(1-\sigma_{m,k}^{i}s_{k,m}^{i^{\prime}%
})V_{i,j}\right\}  +\delta_{m}(V_{i,i+1}-D_{i})\text{.}%
\end{align*}
This is expression in \eqref{vidot} for $j=m$.
\end{proof}

\bigskip

\begin{proof}
[Proof of Proposition \ref{best}]Let
\begin{equation}
{\mathbf{Z}}(t)=\Phi(t,T){\mathbf{Z}}_{0} \label{fond}%
\end{equation}
be the solution of the initial value problem
\[%
\begin{cases}
\dot{\mathbf{Z}}=-\mathcal{A}^{i}(t){\mathbf{Z}}\\
{\mathbf{Z}}(T)={\mathbf{Z}}_{0}\ .
\end{cases}
\]
where $\mathcal{A}^{i}(t)$ is the matrix defined in \eqref{AA}. Consider now a
variation of $\sigma_{i+1,i+2}^{i}(t)$ of the form
\[
\tilde{\sigma}_{i+1,i+2}^{i}(t)=\sigma_{i+1,i+2}^{i}(t)+\delta\eta(t).
\]
where $\eta(t)=0$ for $t>T$ and $\delta$ is a parameter. Differentiating
\eqref{vvidot} with respect to $\delta$ delivers
\[
\partial_{\delta}\dot{\mathbf{V}}^{i}(t)=\rho\partial_{\delta}{\mathbf{V}}%
_{i}-\mathcal{A}^{i}(t)\partial_{\delta}{\mathbf{V}}_{i}-\partial_{\delta
}\mathcal{A}^{i}(t){\mathbf{V}}_{i}.
\]
Using the property that $\partial_{\delta}{\mathbf{V}}_{i}(t)=0$ if $t>T$, the
Duhamel principle gives
\[
\partial_{\delta}{\mathbf{V}}_{i}(t)=\int_{t}^{T}e^{-\rho(\tau-t)}\Phi
(t,\tau)\partial_{\delta}\mathcal{A}^{i}(\tau){\mathbf{V}}_{i}(\tau)d\tau.
\]
From \eqref{AA} it follows that
\[
\partial_{\delta}\mathcal{A}^{i}(\tau){\mathbf{V}}_{i}(\tau)\bigl|_{\delta
=0}=-\alpha\eta(\tau)\Delta_{i+1,i+2}^{i}(\tau)%
\begin{pmatrix}
\sum_{i^{\prime}}p_{i^{\prime},i+2}(\tau)s_{i+2,i+1}^{i^{\prime}}(\tau)\\
0\\
0
\end{pmatrix}
.
\]
Therefore,
\[
\partial_{\delta}{\mathbf{V}}_{i}(t)\bigl|_{\delta=0}=\int_{t}^{\infty}%
\eta(\tau)\Delta_{i+1,i+2}^{i}(\tau){\mathbf{U}}(t,\tau)d\tau\text{,}%
\]
where
\[
{\mathbf{U}}(t,\tau)=e^{-\rho(\tau-t)}\Phi(t,\tau)%
\begin{pmatrix}
\sum_{i^{\prime}}p_{i^{\prime},i+2}(\tau)s_{i+2,i+1}^{i^{\prime}}(\tau)\\
0\\
0
\end{pmatrix}
\text{.}%
\]
Equations \eqref{AA} and \eqref{PF} imply that $\Phi(t,\tau)_{i,j}\geq0$ and
$\Phi(t,\tau)_{i,i}>0$. Moreover, because $p_{i+1,i+2}(\tau)>0$ and
$s_{i+2,i+1}^{i+1}(\tau)=1$, it follows that ${\mathbf{U}}(t,\tau)_{i+1}>0$
and that ${\mathbf{U}}(t,\tau)_{j}\geq0$ for $j=i+2,m$.

Clearly, if $\Delta_{i+1,i+2}^{i}(\tau)\not =0$, the contribution of
$\eta(\tau)$ to the variation $\partial_{\delta}{\mathbf{V}}_{i}(t)>0$ is
different than zero and there is no critical value for $\sigma_{i+1.i+2}%
^{1}(\tau)\in(0,1)$. We can then conclude ${\mathbf{V}}_{i}(t)$ reaches a
maximum at a boundary, i.e. $\sigma_{i+1.i+2}^{1}(\tau)\in\{0,1\}$, and that
\begin{equation}
\sigma_{i+1,i+2}^{i}(\tau)=%
\begin{cases}
1 & \Delta_{i+1,i+2}^{i}(\tau)<0\\
0 & \Delta_{i+1,i+2}^{i}(\tau)>0\,.
\end{cases}
\end{equation}
Finally, as already observed in footnote 8 after Proposition 2, since
$\Delta_{i+1,i+2} ^{i}(t)=0$ for a finite set of switching times, the value of
$\sigma_{i+1,i+2}^{i}(t)$ on such a set does not affect ${\mathbf{V}}_{i}(t)$.
Similar observations hold for $\sigma_{j,k}^{i}$ with $(j,k)\not =(i+1,i+2)$.
This concludes the proof of \eqref{consiste}.
\end{proof}

\begin{proof}
[Proof of Proposition \ref{pro:stab-no-money}]When $M=\delta_{m}=0$ the system
\eqref{pi}-\eqref{pizero} reduces to
\begin{align}
\dot{p}_{1,2}  &  =\alpha\{p_{1,3}[p_{2,1}(1-s_{3,1}^{2})+p_{3,1}%
+p_{3,2}(1-s_{2,3}^{1})]-p_{1,2}p_{2,3}s_{2,3}^{1}\},\label{p1}\\
\dot{p}_{2,3}  &  =\alpha\{p_{2,1}[p_{3,2}(1-s_{1,2}^{3})+p_{1,2}%
+p_{1,3}(1-s_{3,1}^{2})]-p_{2,3}p_{3,1}s_{3,1}^{2}\},\label{p2}\\
\dot{p}_{3,1}  &  =\alpha\{p_{3,2}[p_{1,3}(1-s_{2,3}^{1})+p_{2,3}%
+p_{2,1}(1-s_{1,2}^{3})]-p_{3,1}p_{1,2}s_{1,2}^{3}\}. \label{p3}%
\end{align}
For a given profile of strategies, we need to check that the system
\eqref{p1}-\eqref{p3} has a unique globally attractive steady state.

\textit{Case (0,1,0).} Eq. \eqref{p3} reduces to $\dot{p}_{3,1}=\alpha
(\theta_{3}-p_{3,1})(p_{1,3}+\theta_{2})$, implying that the plane
$p_{3,1}=\theta_{3}$ is globally attractive. On this plane \eqref{p1} reduces
to $\dot{p}_{1,2}=\alpha(\theta_{1}-p_{1,2})\theta_{3}$, implying that the
line $p_{1,2}=\theta_{1}$, $p_{3,1}=\theta_{3}$ is globally attractive. On
this lines \eqref{p2} becomes
\[
\dot{p}_{2,3}=(\theta_{2}-p_{2,3})\theta_{1}-p_{2,3}\theta_{3},
\]
which clearly admits a unique globally attractive fixed point for
$p_{2,3}=\frac{\theta_{2}\theta_{1}}{\theta_{3}+\theta_{1}}$. In brief, under
the profile of strategies (0,1,0), the distribution of inventories converges
globally to the stationary distribution ($\theta_{1},\frac{\theta_{2}%
\theta_{1}}{\theta_{3}+\theta_{1}},\theta_{3}$). For $\theta_{i}=\frac{1}{3}$
this reduces to $\frac{1}{3}(1,\frac{1}{2},1).$ \bigskip

\textit{Case (1,1,0)}. Eq. \ref{p3} becomes $\dot{p}_{3,1}=\alpha(\theta
_{3}-p_{3,1})\theta_{2}$. Consequently, the plane $\theta_{3}=p_{3,1}$ is
globally attractive. The Jacobian, $J$, of the system of the two remaining
eqs. \ref{p1}\ and \ref{p2} on the plane $\theta_{3}=p_{3,1}$ is
\[
J=\alpha\left[
\begin{tabular}
[c]{ll}%
$-(\theta_{3}+p_{2,3})$ & $-p_{1,2}$\\
$(\theta_{2}-p_{2,3})$ & $-(\theta_{3}+p_{1,2})$%
\end{tabular}
\ \ \ \ \right]  .
\]
The determinant of $J$ is always strictly positive, and its trace is always
strictly negative; therefore, both eigenvalues are negative for every relevant
value of $p_{1,2}$ and $p_{2,3}$. Moreover the set $[0,\theta_{1}%
]\times\lbrack0,\theta_{2}]$ is clearly positively invariant and compact. It
thus follows easily from the Poincar\'{e}-Bendixon theorem that the system as
a unique globally attractive fixed point.

To find the stationary distribution, set \ref{p1} and \ref{p2} to zero. They
yield $p_{1,2}=\theta_{1}\theta_{3}/(\theta_{3}+p_{2,3})$ and $p_{1,2}%
=\frac{\theta_{3}}{\theta_{2}/p_{2,3}-1}$, respectively. The two lines
necessarily cross once and only once for $p_{2,3}$ in the interval
$[0$,$\theta_{3}]$. The fixed point is $(p_{1,2}^{\ast},$ $p_{2,3}^{\ast
},\theta_{3})$,where $p_{2,3}^{\ast}=\frac{1}{2}[-(\theta_{1}+\theta
_{3})+\sqrt{(\theta_{1}+\theta_{3})^{2}+4\theta_{1}\theta_{2}}]$ and
$p_{2,1}^{\ast}=\frac{\theta_{1}\theta_{3}}{\theta_{3}+p_{2,3}^{\ast}}$. When
$\theta_{i}=\frac{1}{3}$ the fixed point then is $\frac{1}{3}(\frac{\sqrt{2}%
}{2},\sqrt{2}-1,1)$.\bigskip

\textit{Cases (1,0,1) and (0,1,1)}. A\ Jacobian with similar properties can be
obtained when the profiles of strategies are (1,0,1)\ or (0,1,1). The fixed
point with (1,0,1) is $(p_{1,2}^{\#},\theta_{2},p_{3,1}^{\#})$, where
$p_{1,2}^{\#}=\frac{1}{2}[-(\theta_{3}+\theta_{2})+\sqrt{(\theta_{3}%
+\theta_{2})^{2}+4\theta_{3}\theta_{1}}]$ and $p_{3,1}^{\#}=\frac{\theta
_{1}\theta_{3}}{p_{1,2}^{\#}+\theta_{2}}$. Similarly, under (0,1,1), the fixed
point is ($\theta_{1},\check{p}_{2,3},\check{p}_{3,2}$) where $\check{p}%
_{3,1}=\frac{1}{2}[-(\theta_{2}+\theta_{1})+\sqrt{(\theta_{2}+\theta_{1}%
)^{2}+4\theta_{2}\theta_{3}}]$ and $\check{p}_{2,3}=\frac{\theta_{1}\theta
_{2}}{\check{p}_{3,1}+\theta_{1}}$.

Using a similar argument on can verify that the fixed points supported by the
strategies (0,0,0), (0,0,1), and (1,0,0) are also stable.
\end{proof}

\bigskip

\begin{proof}
[Proof of Proposition \ref{closeup}]For $\epsilon$ sufficiently small,
\[
\Vert{\mathbf{p}}(0)-{\mathbf{p}}^{\ast}\Vert\leq\epsilon\qquad
\mathrm{implies}\qquad\Vert{\mathbf{p}}(t)-{\mathbf{p}}^{\ast}\Vert\leq
C\epsilon\qquad\forall t>0\,.
\]
Since ${\mathbf{p}}^{\ast}$ is a steady state, ${\mathbf{V}}_{i}^{\ast}$ must
be a fixed point of \eqref{vvidot}, that is
\begin{equation}
\rho{\mathbf{V}}_{i}^{\ast}-\mathcal{A}^{\ast}{\mathbf{V}}_{i}^{\ast
}-{\mathbf{v}}^{\ast}=0 \label{vvidotstar}%
\end{equation}
where $\mathcal{A}^{\ast}$ and ${\mathbf{v}}^{\ast}$ are defined in \eqref{AA}
and \eqref{vv} with ${\mathbf{p}}={\mathbf{p}}^{\ast}$. Calling $\delta
{\mathbf{V}}_{i}(t)={\mathbf{V}}_{i}(t)-{\mathbf{V}}^{\ast}$, by subtracting
\eqref{vvidotstar} form \eqref{vvidot} we get
\[
\dot{\delta{\mathbf{V}}_{i}}=\rho\delta{\mathbf{V}}_{i}-\mathcal{A}%
(t)\delta{\mathbf{V}}_{i}-\mathbf{w}(t)
\]
where
\[
\mathbf{w}(t)=\left(  \mathcal{A}(t)-\mathcal{A}^{\ast}\right)  {\mathbf{V}%
}_{i}^{\ast}+({\mathbf{v}}(t)-{\mathbf{v}}^{\ast}).
\]
From Lemma \ref{PFP} below, it follows that
\[
\delta{\mathbf{V}}_{i}(t)={\mathbf{V}}_{i}(t)-{\mathbf{V}}^{\ast}=\int
_{\infty}^{t}e^{\rho(t-s)}\Phi(t,s)\mathbf{w}(s)ds.
\]
Using that $\Vert\mathbf{w}(t)\Vert\leq C\epsilon$ we get that $\Vert
{\mathbf{V}}_{i}(t)-{\mathbf{V}}_{i}^{\ast}\Vert=C\epsilon$ uniformly in $t$.
Since $({\mathbf{p}}^{\ast},{\mathbf{s}}^{\ast})$ form a Nash steady state,
${\mathbf{s}}^{\ast}$ and ${\mathbf{V}}^{\ast}$ must satisfy \eqref{consiste}.
It follows that ${\mathbf{V}}_{i}(t)$ and ${\mathbf{s}}^{\ast}$ still satisfy
\eqref{consiste} if $\epsilon$ is small enough. Thus, if $\Vert{\mathbf{p}%
}(0)-{\mathbf{p}}^{\ast}\Vert$ is small enough,${\mathbf{s}}^{\ast}$ is the
best response to itself for all $t>0$. It follows that $({\mathbf{p}%
}(t),{\mathbf{s}}^{\ast})$ is a Nash equilibrium.
\end{proof}

\bigskip

\begin{proof}
[Proof of Proposition \ref{compu}]To prove the proposition, it is useful to
state the stability properties of \eqref{vidot}. Let ${\mathbf{V}}%
_{i}(t,T,{\mathbf{W}})$ be the solution of \eqref{vidot} obtained by setting
${\mathbf{V}}_{i}(T,T,{\mathbf{W}})={\mathbf{W}}$.

\begin{lemma}
\label{PFP} Given a set of strategies ${\mathbf{s}}(t)$ and
$\boldsymbol{\sigma}(t)$, and a pattern ${\mathbf{p}}(t)$, for any
${\mathbf{W}}_{1}$ and ${\mathbf{W}}_{2}$:
\begin{equation}
\Vert{\mathbf{V}}_{i}(t,T,{\mathbf{W}}_{1})-{\mathbf{V}}_{i}(t,T,{\mathbf{W}%
}_{2})\Vert_{\infty}\leq e^{-\rho(T-t)}\Vert{\mathbf{W}}_{1}-{\mathbf{W}}%
_{2}\Vert_{\infty} \label{Vlimq}%
\end{equation}
where $\Vert{\mathbf{U}}_{i}\Vert_{\infty}=\sup_{j}U_{i,j}$, and $t\leq T$.
Thus the value function at time $t$ of agent $a_{i}$ can be computed as
\begin{equation}
{\mathbf{V}}_{i}(t)=\lim_{T\rightarrow\infty}{\mathbf{V}}_{i}(t,T,{\mathbf{W}%
}) \label{Vlim}%
\end{equation}
where the limit does not depend on ${\mathbf{W}}$ and it is reached
exponentially fast.
\end{lemma}

\begin{proof}
[Proof of Lemma \ref{PFP}]To obtain an explicit representation of
${\mathbf{V}}_{i}(t,T,{\mathbf{W}})$, we apply the Duhamel principle to
\eqref{vvidot} and use the solution $\Phi(t,T)$ of \eqref{fond}. So doing, we
obtain
\begin{equation}
{\mathbf{V}}_{i}(t,T,{\mathbf{W}})=e^{\rho(t-T)}\Phi(t,T){\mathbf{W}}+\int
_{T}^{t}e^{\rho(t-\tau)}\Phi(t,\tau){\mathbf{v}}(\tau)d\tau, \label{Dua}%
\end{equation}
so that
\[
{\mathbf{V}}_{i}(t,T,{\mathbf{W}}_{1})-{\mathbf{V}}_{i}(t,T,{\mathbf{W}}%
_{2})=e^{\rho(t-T)}\Phi(t,T)({\mathbf{W}}_{1}-{\mathbf{W}}_{2})\ .
\]
This implies
\[
\Vert{\mathbf{V}}_{i}(t,T,{\mathbf{W}}_{1})-{\mathbf{V}}_{i}(t,T,{\mathbf{W}%
}_{2})\Vert_{\infty}=e^{\rho(t-T)}\Vert\Phi(t,T)\Vert_{\infty}\Vert
{\mathbf{W}}_{1}-{\mathbf{W}}_{2}\Vert_{\infty}%
\]
The claim in Lemma \ref{PFP} would hold if%
\begin{equation}
\Vert\Phi(t,T)\Vert_{\infty}\leq1\qquad\forall t\leq T\ . \label{contra}%
\end{equation}
To prove \eqref{contra} observe that
\[
\Phi(t,T)=\lim_{N\rightarrow\infty}\prod_{k=N}^{1}\left(  1+\delta
t\mathcal{A}(t_{k})\right)
\]
where
\[
\delta t=\frac{T-t}{N}\qquad t_{k}=t+k\delta t.
\]
Let $\mathcal{C}(t_{K})=1+\delta t\mathcal{A}(t_{k})$. For $\delta t$ small
enough, $\mathcal{C}_{l,j}(t_{k})>0$ for every $l,j$ and
\[
\sum_{j}\mathcal{C}_{l,j}(t_{k})=1\qquad\forall l\text{.}%
\]
Therefore,
\[
\Vert\mathcal{C}(t_{k}){\mathbf{W}}\Vert_{\infty}=\sup_{l}\left\vert \sum
_{j}\mathcal{C}_{l,j}W_{j}\right\vert \leq\sup_{l}\sup_{j}|W_{j}|\sum
_{j}|\mathcal{C}_{l,j}|=\Vert{\mathbf{W}}\Vert_{\infty}\ .
\]
Thus \ref{contra} is verified because
\[
\Vert\Phi(t,T)\Vert_{\infty}\leq\lim_{N\rightarrow\infty}\prod_{k=N}^{1}%
\Vert\mathcal{C}(t_{k})\Vert_{\infty}\leq1\text{.}%
\]
This completes the proof of Lemma 2.
\end{proof}

Returning to Proposition \ref{compu}, observe that $\mathbf{V}_{i}%
(t)=\mathbf{V}_{i}(t,T,{\mathbf{V}}_{i}(T))$. From Lemma \ref{PFP} we get
\[
\Vert{\mathbf{V}}_{i}(t)-{\mathbf{V}}_{i}(t,T,{\mathbf{V}}_{i}^{\ast}%
)\Vert_{\infty}\leq e^{-\rho(T-t)}\Vert{\mathbf{V}}_{i}(T)-{\mathbf{V}}%
_{i}^{\ast}\Vert_{\infty}.
\]
Finally, to go from the infinity distance $\Vert{\mathbf{V}}_{i}%
(t)-{\mathbf{V}}_{i}(t,T,{\mathbf{V}}_{i}^{\ast})\Vert_{\infty}$ to euclidean
distance $\Vert{\mathbf{V}}_{i}(t)-{\mathbf{V}}_{i}(t,T,{\mathbf{V}}_{i}%
^{\ast})\Vert$ we observe that for every ${\mathbf{W}}$ we have $\Vert
{\mathbf{W}}\Vert_{\infty}\leq\|{\mathbf{W}}\|\leq\sqrt{3}\Vert{\mathbf{W}%
}\Vert_{\infty}$. The expression in Proposition \ref{compu} can thus be
obtained as
\begin{align*}
\|{\mathbf{V}}_{i}(t)-{\mathbf{V}}_{i}(t,T,{\mathbf{V}}_{i}^{\ast})\|\leq &
\sqrt{3}\|{\mathbf{V}}_{i}(t)-{\mathbf{V}}_{i}(t,T,{\mathbf{V}}_{i}^{\ast}
)\|_{\infty}\leq\\
\leq &  \sqrt{3}e^{-\rho(T-t)}\|{\mathbf{V}}_{i}(T)-{\mathbf{V}}_{i}^{\ast}
\|_{\infty}\leq\sqrt{3}e^{-\rho(T-t)}\|{\mathbf{V}}_{i}(T)-{\mathbf{V}}%
_{i}^{\ast}\|\,.
\end{align*}

\end{proof}

\section{Nash Equilibria and the Stable Manifold Theorem}

\label{app:manifold}The iteration procedure to find Nash equilibria is similar
to that used by Perron to prove the stable manifold theorem (for an
illustration see, among others, Robinson, 1995). Here, we discuss similarities
and differences. Consider the system of differential equations
\begin{equation}%
\begin{cases}
\dot{x}=-\lambda x+f^{-}(x,y)\\
\dot{y}=\mu y+f^{+}(x,y)\\
\end{cases}
\label{eq}%
\end{equation}
where $x\in\mathbb{R}^{n}$, $y\in\mathbb{R}^{m}$, $\lambda,\mu>0$ and
\[
\lim_{|x|+|y|\rightarrow0}\frac{|f^{-}(x,y)|+|f^{+}(x,y)|}{|x|+|y|}=0.
\]
The system \eqref{eq} is the sum of linear ($-\lambda x$ and $\mu y$) and
non-linear terms ($f^{+}$ and $f^{-}$); its fixed point is $(x,y)=(0,0)$.

The stable manifold theorem states that for every $x_{0}\in\mathbb{R}^{n}$,
with $|x_{0}|$ sufficiently small, there is unique $y_{0}$ such that the
solution $(x(t),y(t))$ of \eqref{eq} starting at $(x_{0},y_{0})$ satisfies
\begin{equation}
\label{lim}\lim_{t\rightarrow\infty}(x(t),y(t))=(0,0).
\end{equation}
Moreover, it says that the point $y_{0}$ is given by a smooth function of
$x_{0}$, that is $y_{0}=W^{-}(x_{0})$. The graph of $W^{-}$, that is the set
$(x_{0},W^{-}(x_{0}))$, is called the local stable manifold of $(0,0)$.
Perron's proof is based on the following representation of the solution of
\eqref{eq}:
\begin{align}
x(t)  &  =e^{-\lambda t}x_{0}+\int_{0}^{t}e^{-\lambda(t-\tau)}f^{-}%
(x(\tau),y(\tau))d\tau\label{stable}\\
y(t)  &  =\int_{\infty}^{t}e^{\mu(t-\tau)}f^{+}(x(\tau),y(\tau))d\tau.
\label{unstable}%
\end{align}
The proof starts from a guess $x^{0}(t)=x_{0}e^{-\lambda t}$ and $y^{0}(t)=0$
for all $t>0$. It then computes a new approximation for the evolution of the
stable variable $x(t)$, $x^{1}(t)$, through \eqref{stable}. Inserting $x^{1}$
and $y^{0}$ into \eqref{unstable} yields an approximation for the unstable
variable $y^{1}(t)$. Note that while in \eqref{stable} time runs forward
($\tau$ goes from $0$ to $t$), in \eqref{unstable} it runs backward (in $\tau$
goes from $\infty$ to $t$). Therefore both integrations are stable. Iterating
the two steps just described yields a sequence $(x^{n}(t),y^{n}(t))$ that
approximates the solution \eqref{stable}-\eqref{unstable}. Because the
exponential factors in the integrals of \eqref{stable} and \eqref{unstable}
have negative exponents, if $x_{0}$ is sufficiently small, the map from
$(x^{n},y^{n})$ to $(x^{n+1},y^{n+1})$ is a contraction. Finally, the Banach
fixed-point theorem guarantees that the sequence $(x^{n}(t),y^{n}(t))$
converges uniformly to a solution $(x(t),y(t))$ of \eqref{eq} with
$x(0)=x_{0}$ and satisfying \eqref{lim}.

There are similarities between Perron's approach in proving the manifold
theorem and the construction of Nash equilibria discussed in Section 4. First,
the Nash steady state $({\mathbf{p}}^{\ast},{\mathbf{s}}^{\ast})$ corresponds
to the fixed point $(0,0)$ of \eqref{eq}. Second, the ${\mathbf{p}}$ in
\eqref{pi}-\eqref{pizero} is comparable to the $x$ in \eqref{eq}. Third, in
\eqref{Dua} the discount rate $\rho$ plays the same role of the unstable
exponent $\mu$ in \eqref{unstable}. But there are also important differences
because. Fourth the value function ${\mathbf{V}}^{i}$ in \eqref{vidot} is
somewhat comparable to $y$ in \eqref{eq}. But there is also an important
difference: The value function ${\mathbf{V}}^{i}$ affects the evolution of
${\mathbf{p}}$ through the intermediation of the strategies ${\mathbf{s}}$. In
addition, differently from \eqref{eq}, the procedure we presented in Section 4
does not split the evolution of $({\mathbf{p(}}t),{\mathbf{s(}}t))$ near
$({\mathbf{p}}^{\ast},{\mathbf{s}}^{\ast})$ between a linear and a nonlinear
part. Therefore, it is better suited to follow the dynamics away form the
steady state.

Finally, as noted in Section 4, it is important to recognize that the
evolution of the distribution of inventories, ${\mathbf{p}}(t)$, and that of
the value functions ${\mathbf{V}}(t)$, can be studied jointly. Starting from a
${\mathbf{p}}(T)$ close to the steady state ${\mathbf{p}}^{\ast}$, with a
${\mathbf{V}}_{i}(T)={\mathbf{V}}_{i}^{\ast}$ and $\boldsymbol{\sigma}%
_{i}(T)=\boldsymbol{\sigma}_{i}^{\ast}$, one may integrate \eqref{vidot} and
\eqref{pi}-\eqref{pizero} backward in time -- $t$ goes from $T$ to $0$. To
find the approximate solution of the Nash equilibrium it would suffice to
alter $\boldsymbol{\sigma}_{i}(t)$ along the integration process so as to be
consistent with the value functions ${\mathbf{V}}_{i}(t)$ -- i.e. to satisfy
\eqref{consiste}. While this procedure usually works well for low-dimensional
systems (see, for instance, Brunner and Strulik, 2002), it presents
limitations for the type of research question we are after. Our objective is
to obtain a pattern $({\mathbf{s}}(t),{\mathbf{p}}(t))$ that goes through any
initial conditions, that is, through any arbitrary points in the space of the
distribution of inventories, ${\mathbf{p}}(0)$. As the dimension of the
manifold expands, guiding the system toward a particular point on the state
space by integrating \eqref{vidot} and \eqref{pi}-\eqref{pizero} backward in
time is challenging because some regions of the manifold may be hard to reach.
Conversely, the method proposed here offers total control over the initial
condition, an indispensable feature for running macroeconomic experiments.

\newpage%

\begin{table}[tbp]%
\caption{Baseline Parameters}%
\medskip\hspace*{\fill}%
\begin{tabular}
[c]{ccccccccc}\hline\hline
Model & \textbf{Discount} & \textbf{Matching} &
\multicolumn{3}{c}{\textbf{Utility }} & \multicolumn{3}{c}{\textbf{Storage
Costs}}\\\hline
& $\delta$ & $\alpha$ &  & $u_{i}$ & $D_{i}$ & $c_{1}$ & $c_{2}$ & $c_{3}%
$\\\hline
A & 0.03 & 1 &  & 1 & 0.028 & 0.03 & 0.1 & 0.2\\\hline
B & 0.03 & 1 &  & 1 & 0.028 & 0.1 & 0.05 & 0.03\\\hline
\end{tabular}
\hspace*{\fill}

\label{parameters}%
\end{table}%
\bigskip

\bigskip%

\begin{figure}[tbp]%
\caption{Overview of Equilibria}\hspace*{\fill}%
\begin{tabular}
[c]{c}%
\\%
{\includegraphics[
height=3.9186in,
width=5.066in
]%
{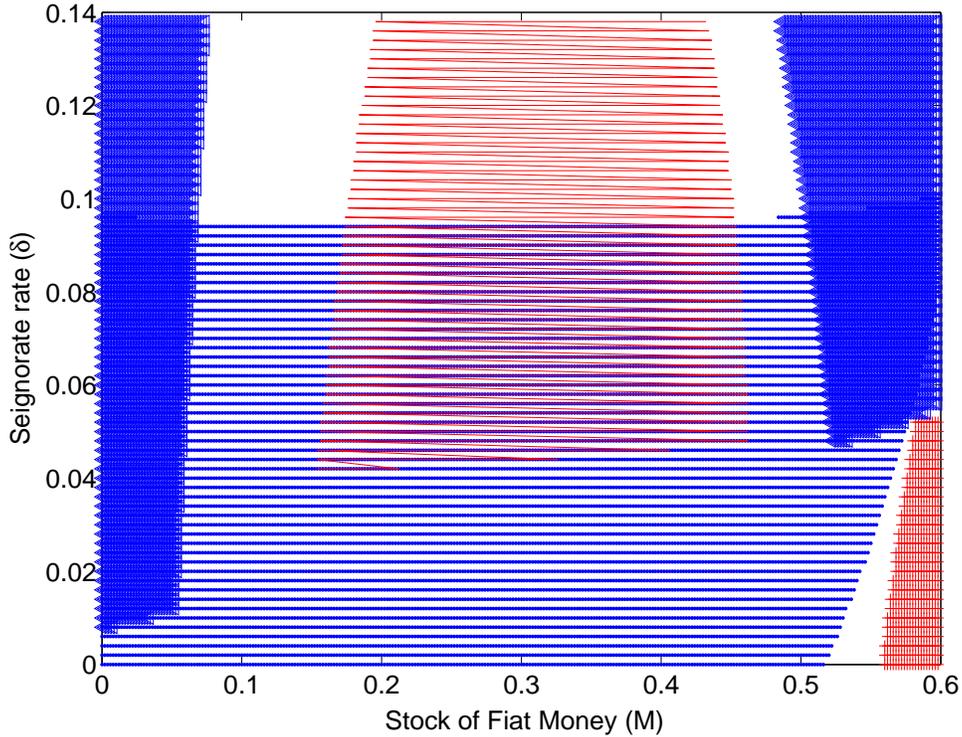}%
}
\end{tabular}
\hspace*{\fill}

{\small - Note. There are four equilibria. Sometimes two equilibria coexist.
The\ dark dotted area in the lower-middle part of the figure represents
}$\mathbf{s}_{3}=${\small (1,1,0) and full acceptance of fiat money. The + in
light color on the lower-right side denotes }$\mathbf{s}_{3}=${\small (0,1,0)
and full acceptance of fiat money. The remaining two types of equilibria are
}$\mathbf{s}_{3}={\small (1,1,0)}${\small (the sign -- in light color, in the
middle-upper part) and }$\mathbf{s}_{3}={\small (0,1,0)}${\small \ (the dark
sign }$<${\small on the top-right and top-left region) with }$s_{1,m}^{2}=0$
{\small and }$s_{j,m}^{i}=1${\small when }$i\neq2${\small and }$j\neq
1${\small . The population is equally divided between the three types
(}$\theta_{i}=\frac{1}{3}${\small ). The remaining parameters values are in
table \ref{parameters}, Model A.}\label{onethird}%
\end{figure}%
%

\begin{figure}[tbp]%
\caption{Acceptance of Commodity and Fiat Money}\hspace*{\fill}%
\begin{tabular}
[c]{c}%
Panel A, Phase Diagram\\%
{\includegraphics[
height=4.1154in,
width=5.3391in
]%
{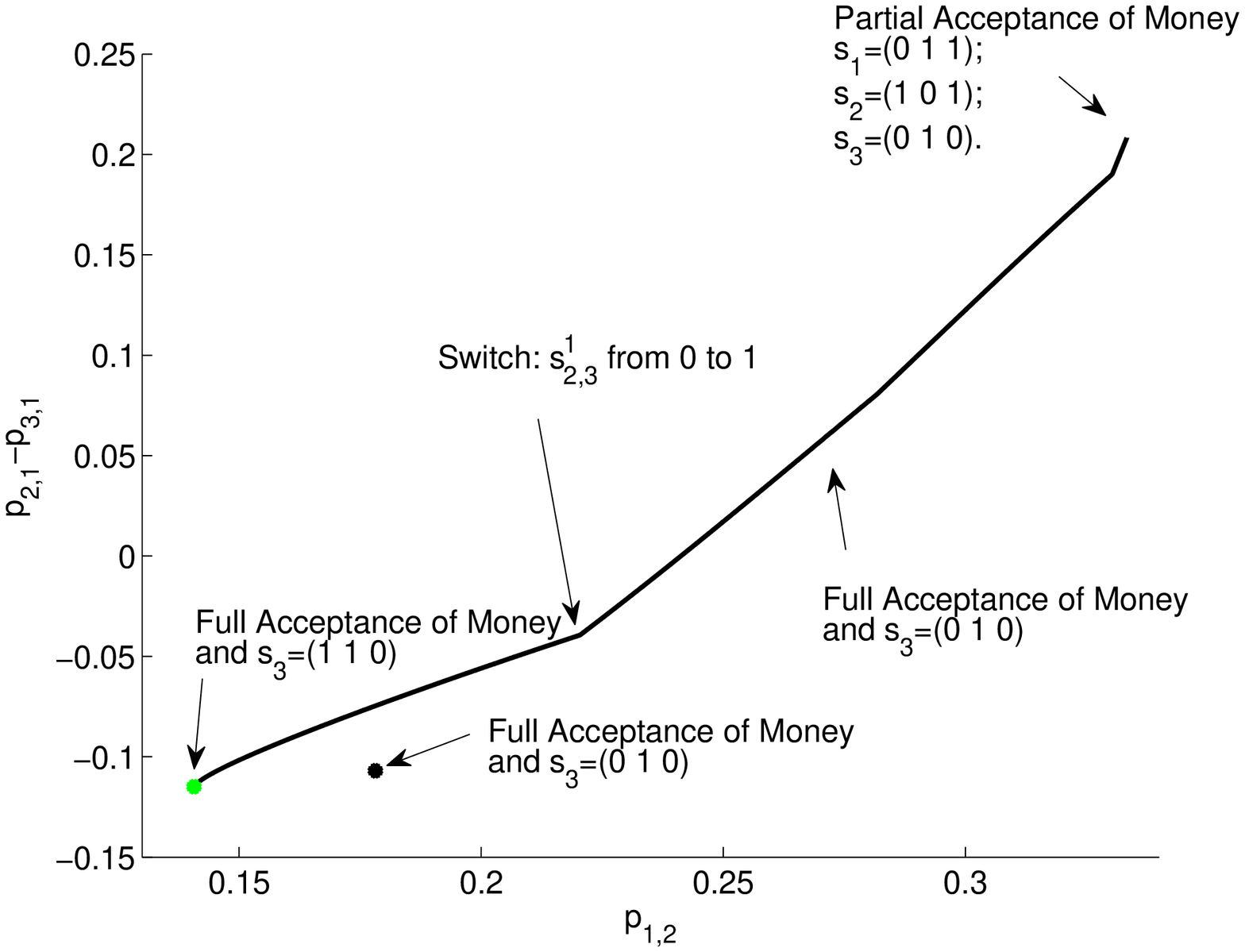}%
}
\\
\\%
\begin{tabular}
[c]{cc}%
Panel B,\ Strategies: $\delta_{m}=0.07$ & Panel C, Strategies: $\delta
_{m}=0.06$\\%
{\includegraphics[
height=2.4682in,
width=3.0685in
]%
{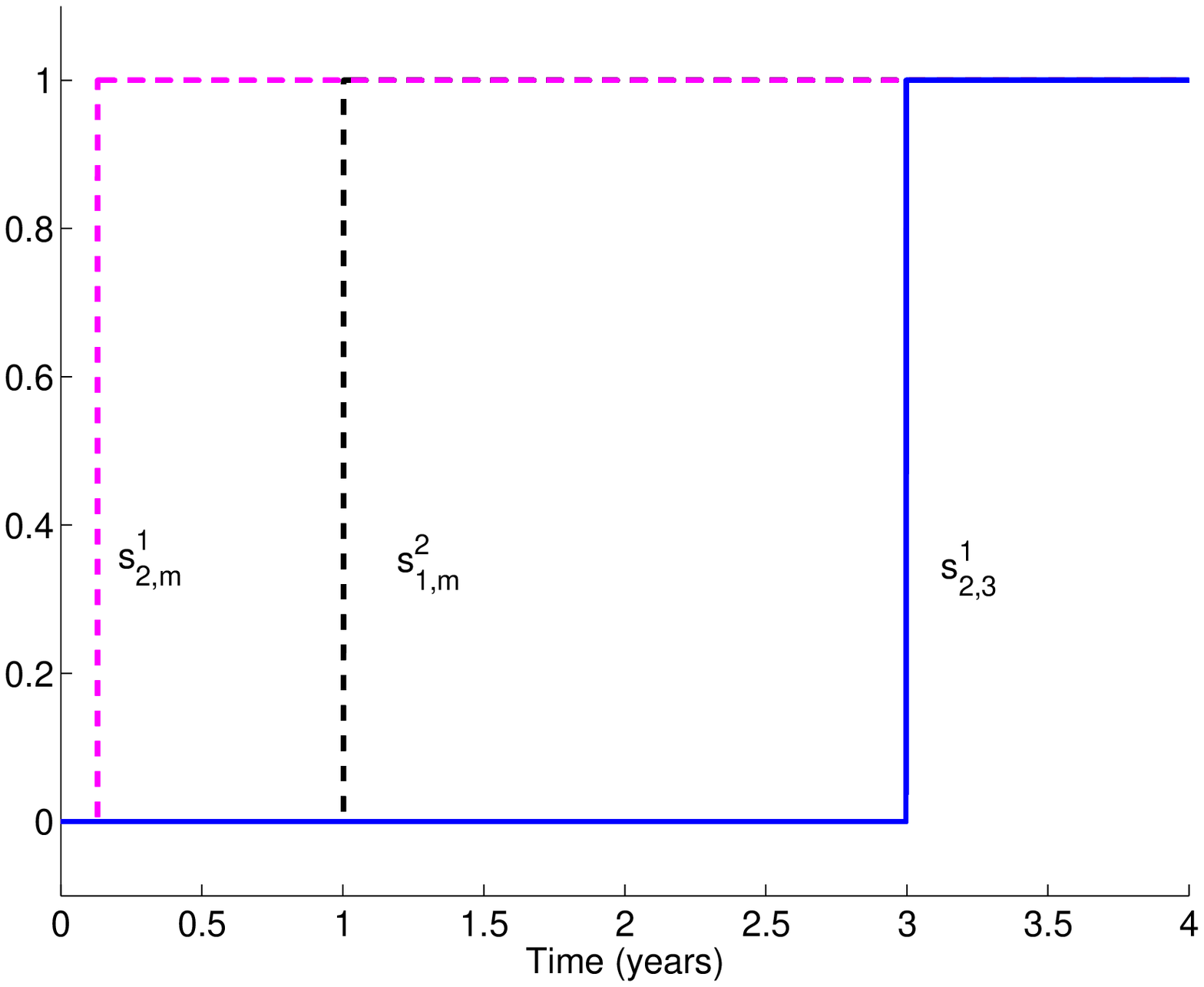}%
}
&
{\includegraphics[
height=2.4657in,
width=3.066in
]%
{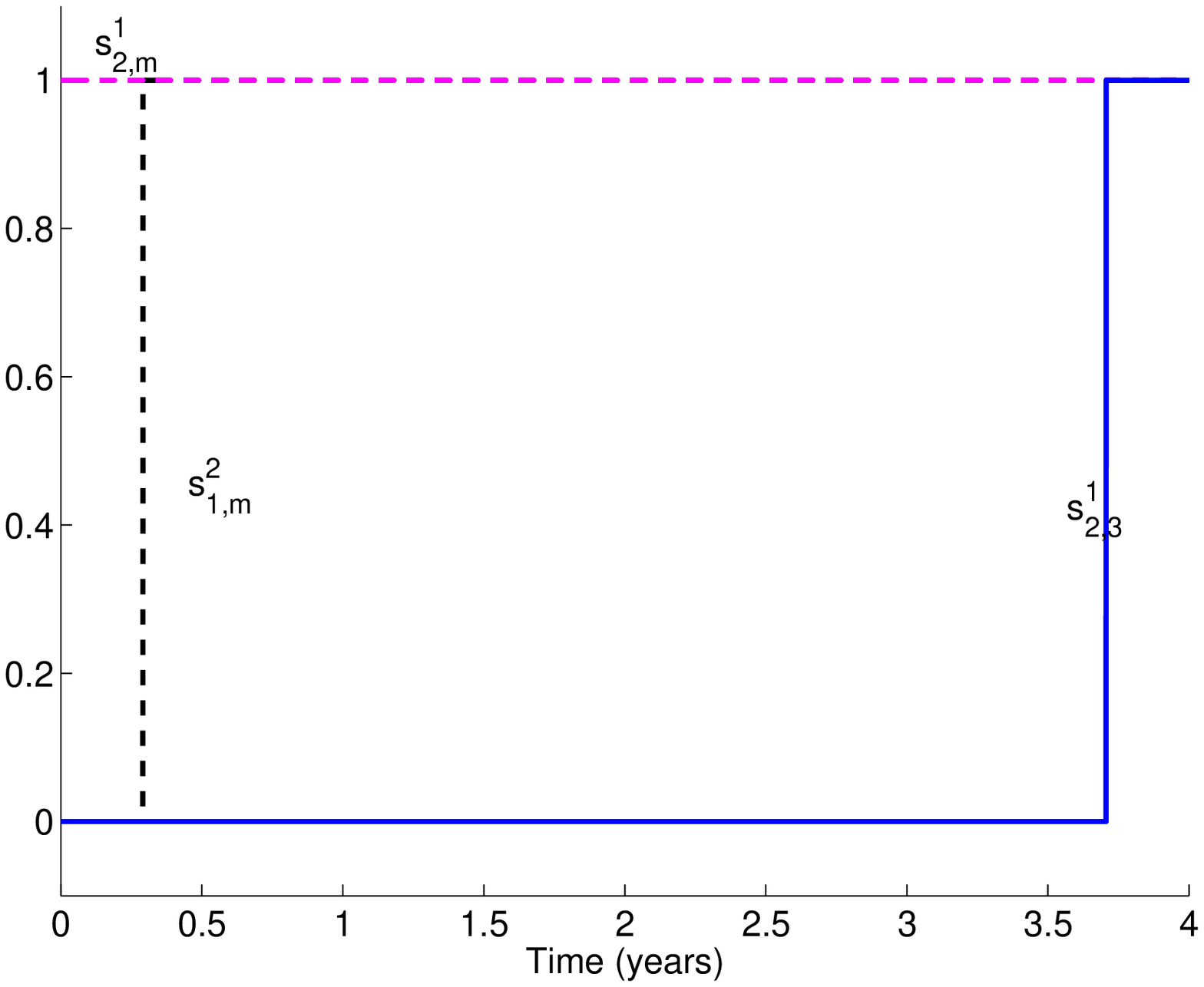}%
}
\end{tabular}
\end{tabular}
\hspace*{\fill}

{\small - Note: Panels A illustrates the convergence to a }$\mathbf{s}_{3}%
=${\small (1,1,0) steady state equilibrium with full acceptance of money. The
population is equally split between the three types (}$\theta_{i}=\frac{1}{3}%
$). {\small The initial condition }$\mathbf{p}(0)=(\theta_{1},0,0,0,\frac
{M}{4})${\small . Panels B\ and C show the switch of }$s_{2,m}^{1}${\small ,
}$s_{1,m}^{2}${\small and of }$s_{2,3}^{1}${\small  from }$0$ {\small to }%
$1${\small . }\label{double_emergence}%
\end{figure}%
%

\begin{figure}[tbp]%
\caption{Reduction of Seignorage}\hspace*{\fill}%
\begin{tabular}
[c]{cc}%
Panel A:\ Phase Diagram & Panel B: Production\\%
{\includegraphics[
height=2.3541in,
width=3.0414in
]%
{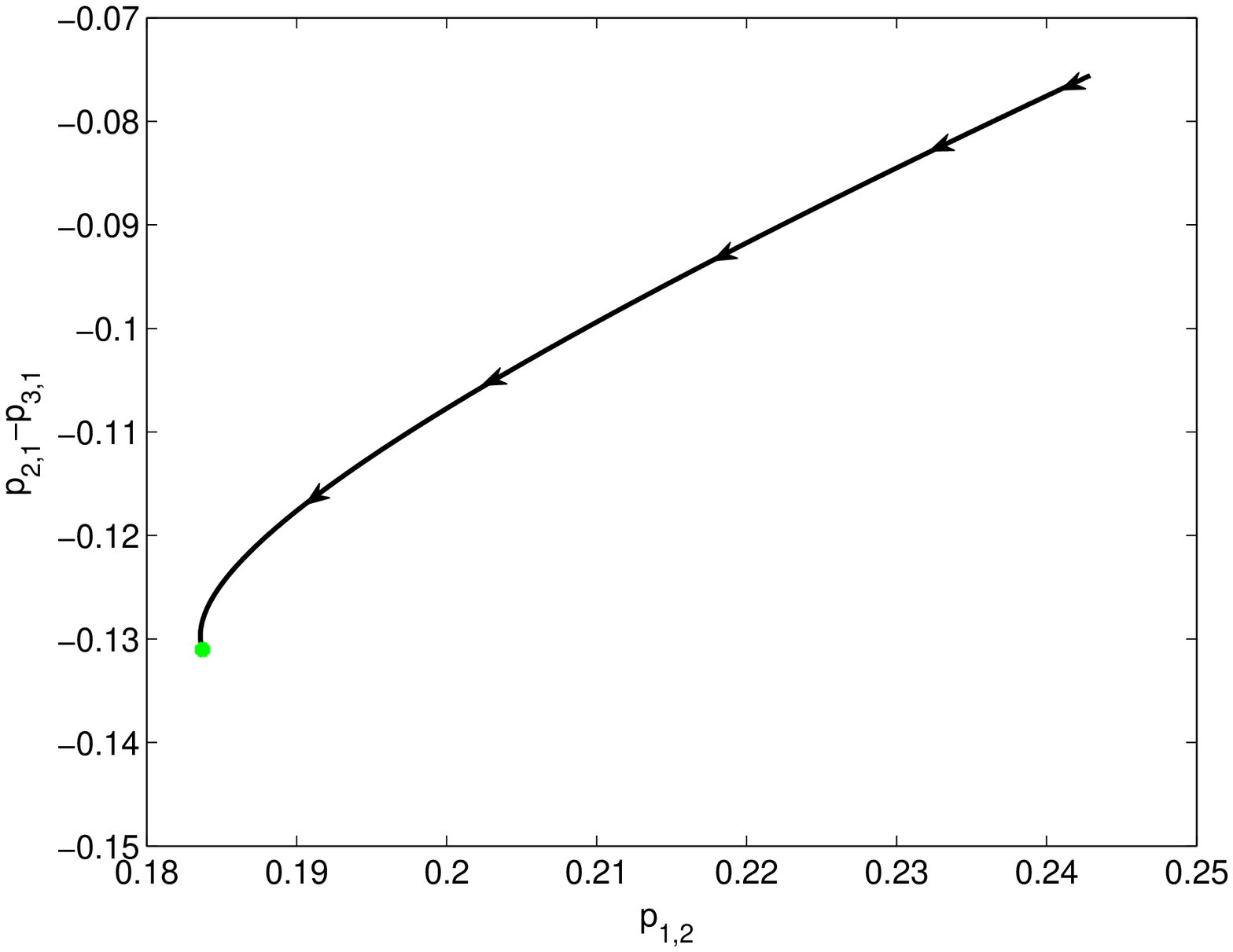}%
}
&
{\includegraphics[
height=2.3025in,
width=2.8947in
]%
{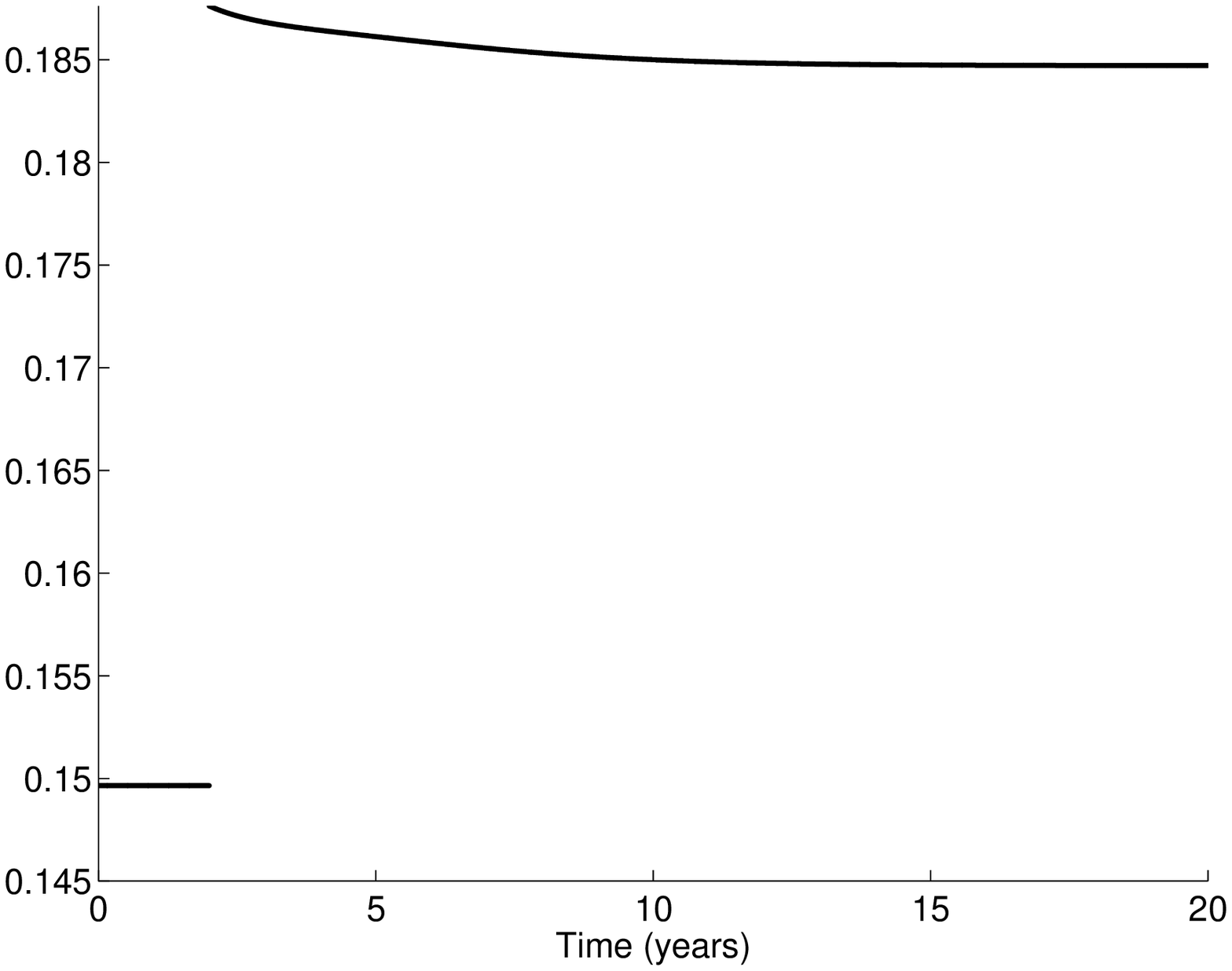}%
}
\end{tabular}
\hspace*{\fill}

{\small - Note: The seignorage rate, }$\delta_{m}${\small , goes from }%
$0.1${\small  to  }$0.02${\small . The economy transits from a steady state
equilibrium in which }$\mathbf{s}_{1}=${\small (1,1,1), }$\mathbf{s}_{2}%
=${\small (0,1,1), and }$\mathbf{s}_{3}=${\small (0,1,0), to new equilibrium
in which }$\mathbf{s}_{1}=\mathbf{s}_{2}=${\small (1,1,1), and }%
$\mathbf{s}_{3}=${\small (1,1,0). The stock of fiat money is }$M${\small =0.3
and }$\theta_{i}=\frac{1}{3}${\small . The remaining parameters are in table
\ref{parameters}, Model A. }\label{reduction_seignorage}%
\end{figure}%
%

\begin{figure}[tbp]%
\caption{Overview Steady State Equilibria, Model A}\hspace*{\fill}%
\begin{tabular}
[c]{c}%
Panel A: $\delta_{m}=0.02$\\%
{\includegraphics[
height=3.7738in,
width=4.6077in
]%
{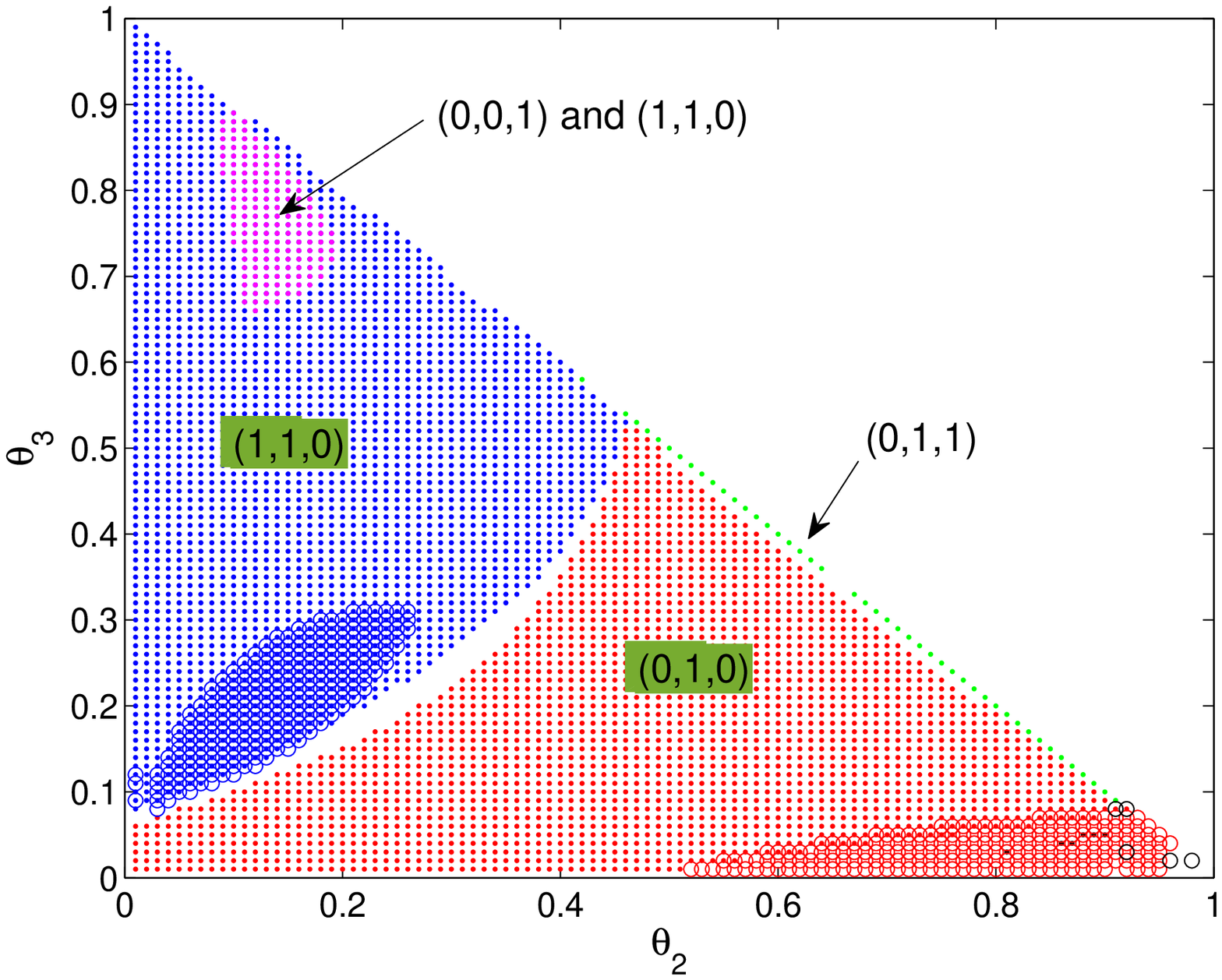}%
}
\\
Panel B: $\delta_{m}=0.10$\\%
{\includegraphics[
height=3.4629in,
width=4.3504in
]%
{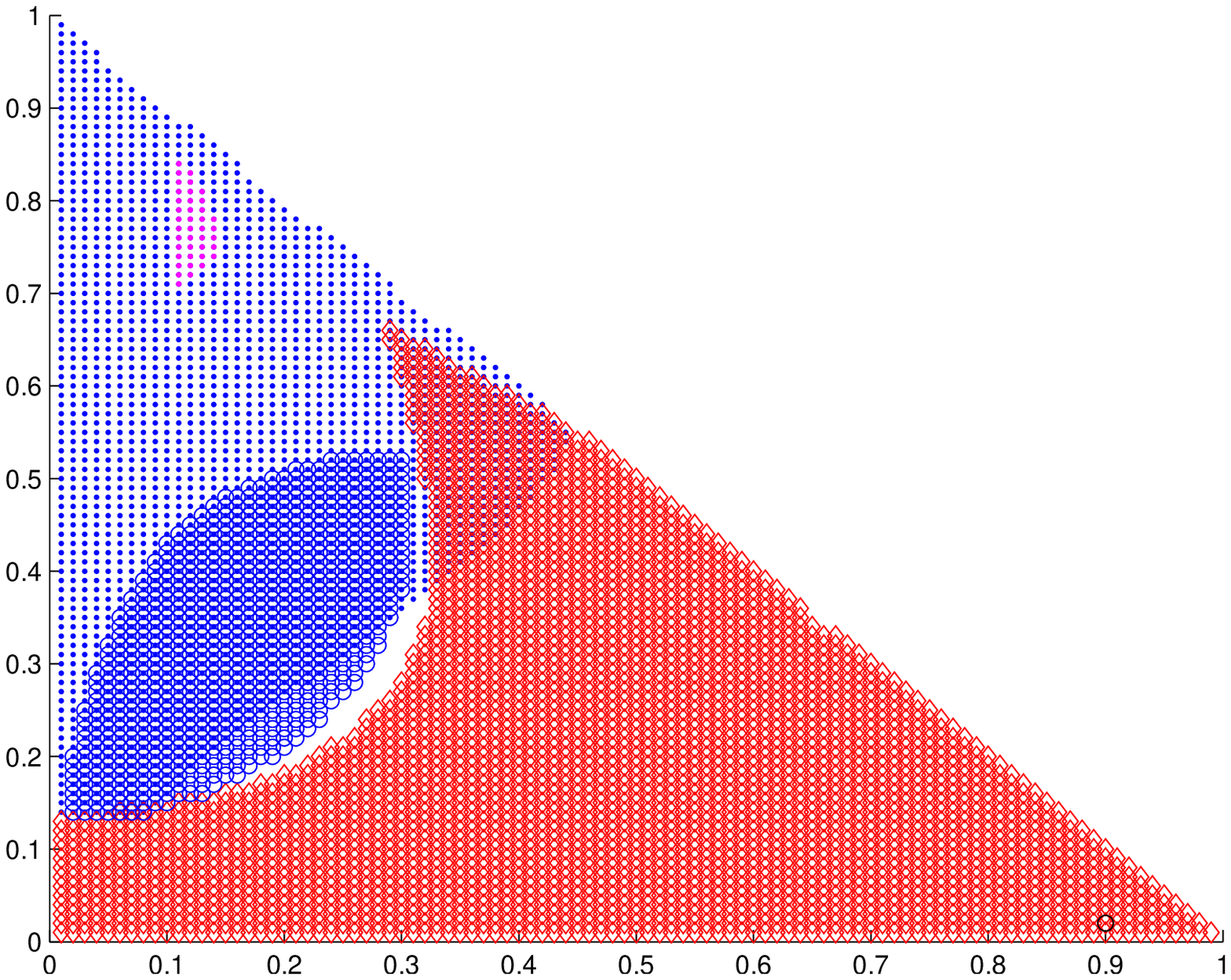}%
}
\end{tabular}
\hspace*{\fill}

{\small - Note:\ On equilibria represented by a dot, }$s_{j,m}^{i}=1${\small
and :\ }$\mathbf{s}_{3}={\small (0,1,0)}${\small (red dot); }$\mathbf{s}%
_{3}={\small (1,1,0)}${\small (blue dot)}; $\mathbf{s}_{3}={\small (0,0,1)}$
{\small (magenta dot); }$\mathbf{s}_{3}={\small (0,1,1)}$ {\small (green
dot}){\small . On equilibria represented by a circle, }$s_{1,m}^{2}%
=0${\small \ and: }$\mathbf{s}_{3}={\small (0,1,0)}${\small (red
circle}){\small ; }$\mathbf{s}_{3}={\small (1,1,0)}${\small (blue
circle}){\small , and }$\mathbf{s}_{3}={\small (1,1,0)}${\small (black
circle}){\small . The triplets in parenthesis in plot A denote }%
$\mathbf{s}_{3}$. {\small M=0.3 in both plots. For the remaining parameters
values see table \ref{parameters}, Model A.}\label{monetary_equilibria}%
\end{figure}%
\bigskip%

\begin{figure}[tbp]%
\caption{Multiple Equilibria}\hspace*{\fill}%
\begin{tabular}
[c]{c}%
{\includegraphics[
height=3.9887in,
width=5.2498in
]%
{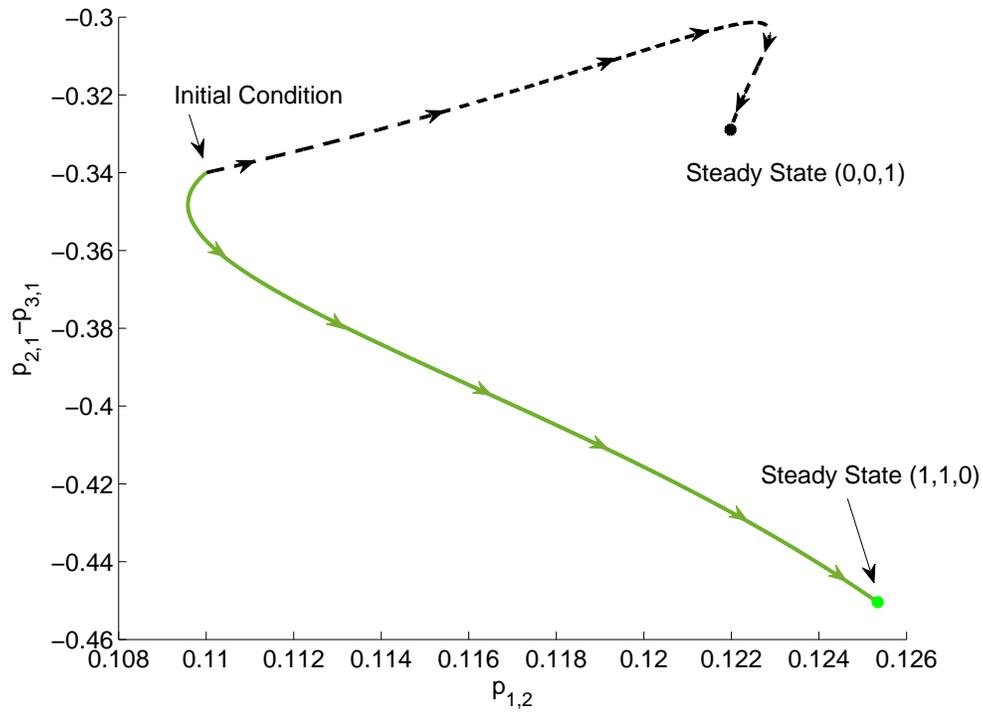}%
}
\end{tabular}
\hspace*{\fill}

{\small - Note: The distribution of the population is as follows:\ }%
$\theta_{1}=0.24${\small , }$\theta_{2}=0.16${\small , and }$\theta_{3}=0.7$.
{\small The rate of seignorage and the stock of fiat money is 0.05 and 0.2,
respectively. The storage costs are: }$c_{1}=0.03${\small , }$c_{2}%
=0.1${\small , and }$c_{3}=0.2${\small . The initial condition is }%
$\mathbf{p}(0)=(0.11,0.06,0.38,0.02,0.06)$. \label{multiple_equilibria}%
\end{figure}%

\bigskip%
\begin{figure}[tbp]%
\caption{Overview Steady State Equilibria, Model B}\hspace*{\fill}%
\begin{tabular}
[c]{c}%
Panel A: $\delta_{m}=0.02$\\%
{\includegraphics[
height=3.4952in,
width=4.3181in
]%
{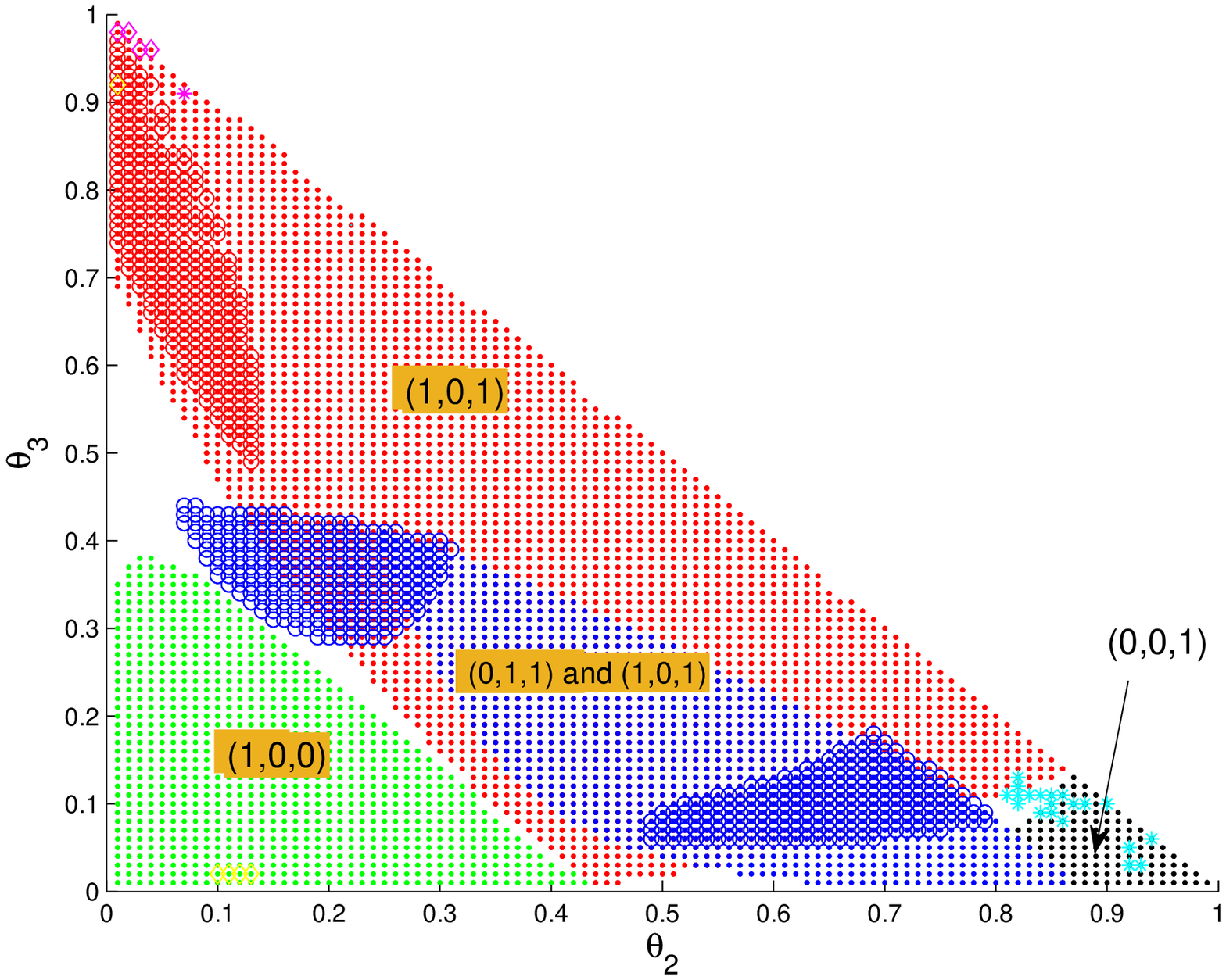}%
}
\\
Panel B: $\delta_{m}=0.1$\\%
{\includegraphics[
height=3.3569in,
width=4.4168in
]%
{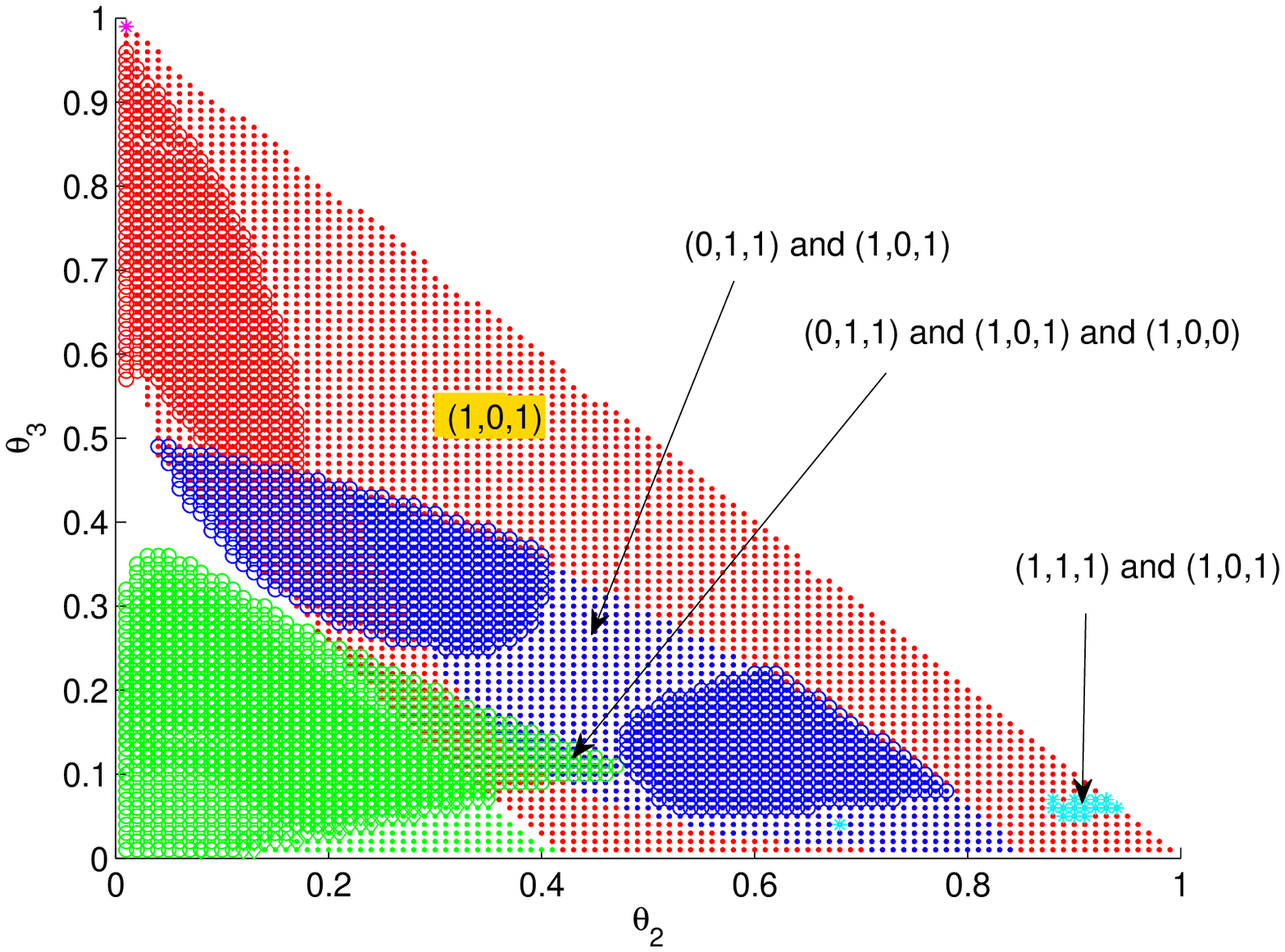}%
}
\end{tabular}
\hspace*{\fill}

{\small - Note: On equilibria represented by a dot, }$s_{j,m}^{i}=1${\small
and: }$\mathbf{s}_{3}={\small (1,0,1)}${\small (red dot); }$\mathbf{s}%
_{3}={\small (0,1,0)}${\small (blue dot); }$\mathbf{s}_{3}={\small (1,0,0)}%
${\small (green dot); and} $\mathbf{s}_{3}={\small (0,0,1)}$ {\small (black
dot). On equilibria represented by a circle, }$s_{1,m}^{2}=0${\small \ and:
}$\mathbf{s}_{3}={\small (1,0,1)}${\small (red circle}){\small ; }%
$\mathbf{s}_{3}={\small (0,1,0)}${\small (blue circle}){\small , }%
$\mathbf{s}_{3}={\small (1,0,0)}${\small (green circle}); {\small and
}$\mathbf{s}_{3}={\small (0,0,1)}$ {\small (black circle}){\small . The
triplets in parenthesis in the two plots denote }$\mathbf{s}_{3}${\small .
M=0.3 on both plots. For other parameters values see table \ref{parameters},
Model B. }\label{overview_equilibria_B}%
\end{figure}%
%

\begin{figure}[tbp]%
\caption{Acceptance of Commodity and Fiat Money, Model B}\hspace*{\fill}%
\begin{tabular}
[c]{c}%
\raisebox{-0pt}{\includegraphics[
height=3.8642in,
width=5.2765in
]%
{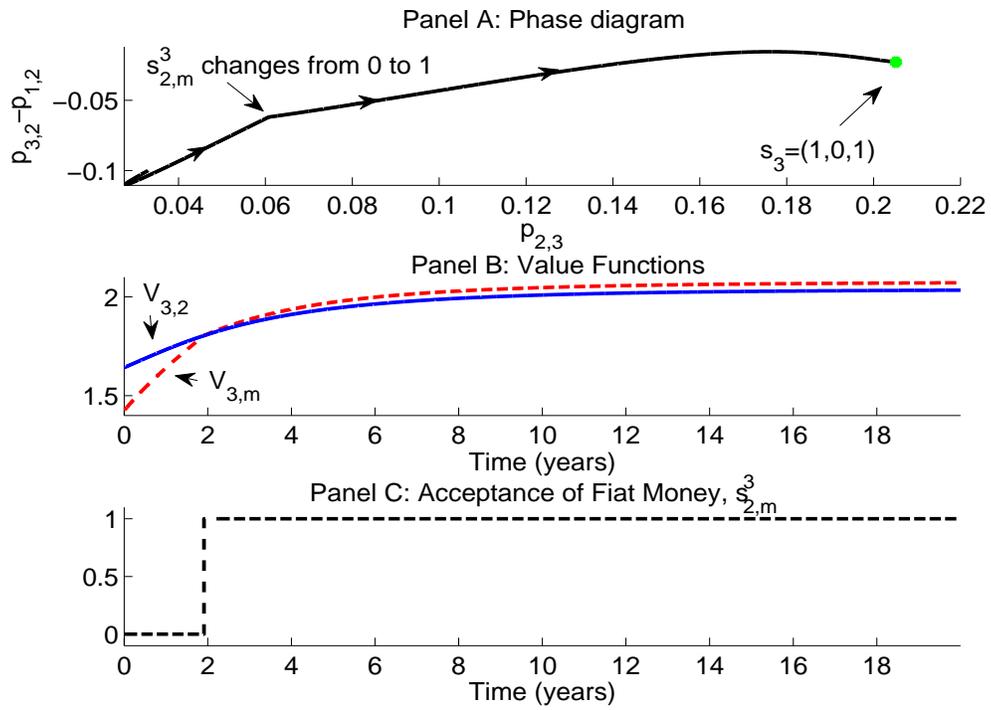}%
}
\end{tabular}
\hspace*{\fill}

{\small -- Note:\ The parameters' values are:\ }$\delta_{m}=${\small 0.1, }%
$M${\small =0.3, and }$\theta_{i}=\frac{1}{3}${\small . See table
(\ref{parameters}), Model B, for remaining parameters' values. The initial
condition is }$\mathbf{p}(0)=(\frac{1}{2}\theta_{1},\frac{1}{10}\theta
_{2},\theta_{3},M,0)$.\label{emergencemodelB}%
\end{figure}%
%

\begin{figure}[tbp]%
\caption{Welfare, Fiat Money, and Seignorage}\hspace*{\fill}%
\begin{tabular}
[c]{c}%
Panel A: Fiat Money\\%
{\includegraphics[
height=3.6179in,
width=4.3624in
]%
{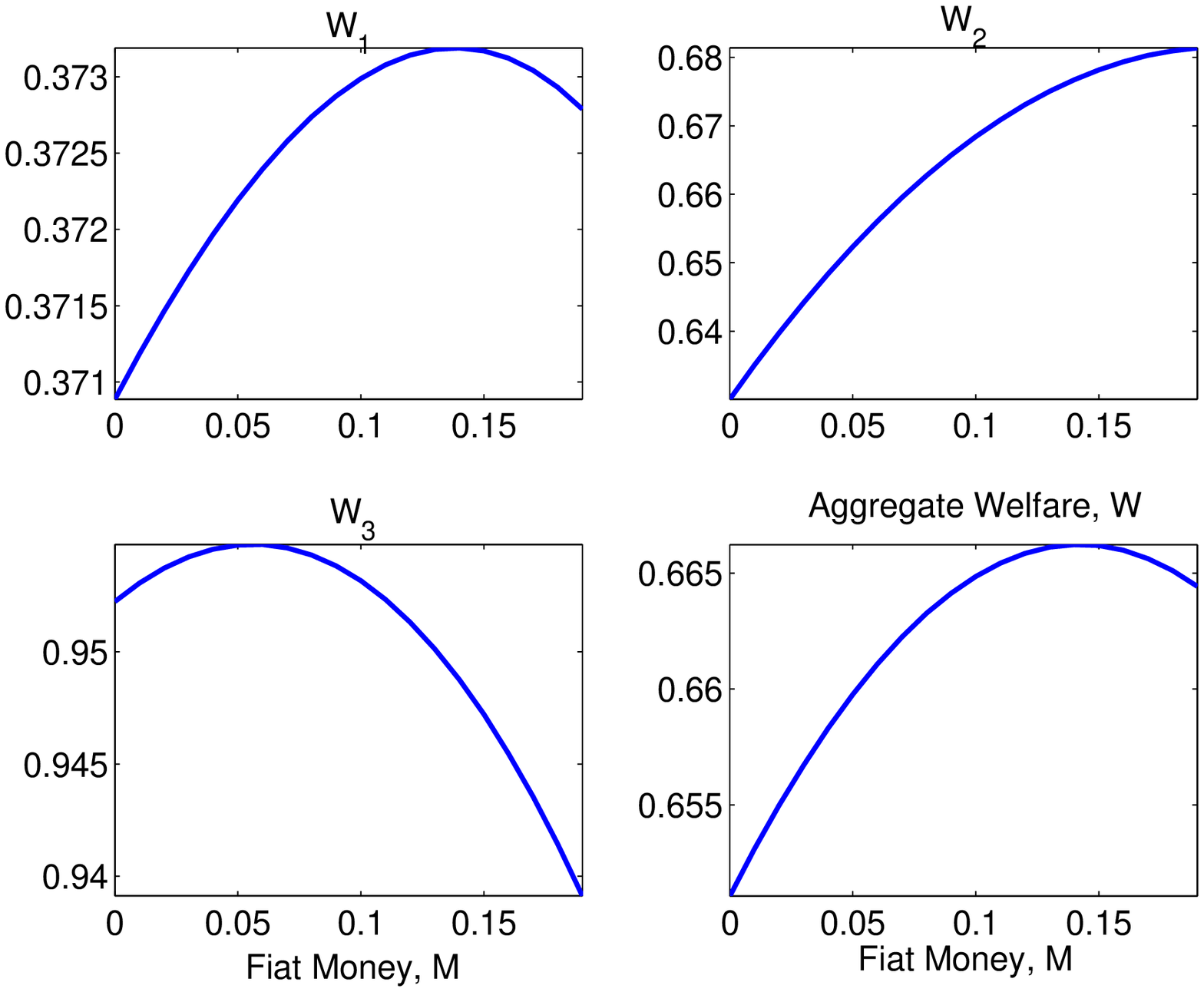}%
}
\\
Panel B: Seigniorage\\%
{\includegraphics[
height=3.8744in,
width=4.6815in
]%
{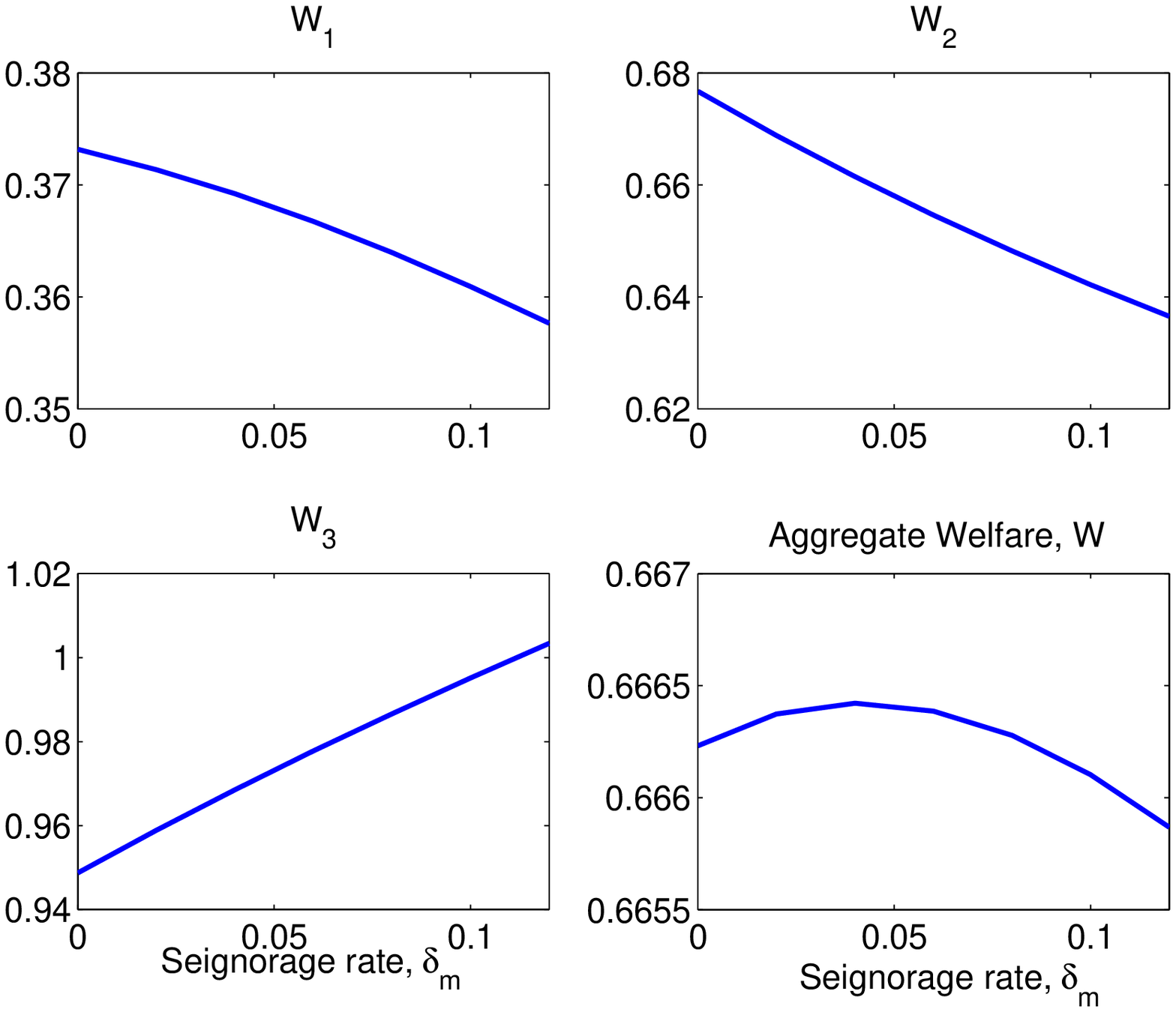}%
}
\end{tabular}
\hspace*{\fill}

{\small Note:\ The value of }$\delta_{m}$ {\small and }$M$ {\small that
maximizes the society's welfare is }$0.04$ {\small and }$0.14${\small ,
respectively. The population is equally split between the three types
(}$\theta_{i}=\frac{1}{3}$){\small . For remaining parameters see table
(\ref{parameters}), Model A.}

\label{welfareM}%
\end{figure}%

\bigskip\bigskip

\end{document}